\documentclass[a4paper,12pt]{amsart} 

\usepackage{graphicx}
\usepackage{algorithm}
\usepackage{algorithmic}
\usepackage{fullpage}
\usepackage{verbatim}
\usepackage{subfigure}
\numberwithin{equation}{section}

\newtheorem{theorem}{Theorem}[section]

\newtheorem{lemma}[theorem]{Lemma}

\theoremstyle{definition}

\newtheorem{remark}[theorem]{Remark}

\newcommand{\rset}{\mathbb{R}}      %
\newcommand{\tpsi}{\psi} %

\def\LPROD#1#2{\langle {#1},{#2}\rangle}

\def\BO{0}

\def\HOPER{\hat{H}}
\def\VOPER{{V}}
\def\SOPER{\mathcal{S}}
\def\ROPER{\mathcal{R}}

\def\BTOPER{{\tilde{\mathcal{B}}}}

\def\PSHARP#1{\Pi{#1}}
\def\PSHARPT#1#2{\Pi({#1}){#2}}

\def\Oo{\mathcal{O}}

\def\COMMA{\,,}             
\def\PERIOD{\,.}            
\def\SEP{{\,|\,}}           

\def\Iunit{i}

\def\BIGO{\Oo}

\begin{document}

 \author{Ashraful Kadir}
 \address{Department of Mathematics,
    Kungl. Tekniska H\"ogskolan,
   100 44 Stockholm,
   Sweden}
 \email{smakadir@csc.kth.se}

 \author{Mattias Sandberg}
 \address{Department of Mathematics,
   Kungl. Tekniska H\"ogskolan,
   100 44 Stockholm,
   Sweden}
 \email{msandb@kth.se}

 \author{Anders Szepessy}
 \address{Department of Mathematics,
   Kungl. Tekniska H\"ogskolan,
   100 44 Stockholm,
   Sweden}
 \email{szepessy@kth.se}
%

%
\title[Car-Parrinello and  Ehrenfest MDs with adaptive mass]{An adaptive mass algorithm for Car-Parrinello and Ehrenfest { ab initio} molecular dynamics}

%

\begin{abstract}
Ehrenfest and Car-Parrinello molecular dynamics are computational alternatives to approximate Born-Oppenheimer
molecular dynamics without solving the electron eigenvalue problem at each time-step.
A non-trivial issue is to choose the artificial electron mass parameter appearing in the Car-Parrinello method to achieve  both good accuracy and high computational efficiency.
In this paper, we propose an algorithm, motivated by the Landau-Zener probability, 
to systematically choose an
artificial mass dynamically, 
 which makes the Car-Parrinello and Ehrenfest molecular dynamics methods dependent
only on the problem data. Numerical experiments for simple model problems show that the time-dependent adaptive 
artificial mass parameter improves the efficiency of the Car-Parrinello and Ehrenfest molecular dynamics.
\end{abstract}

\maketitle
\tableofcontents

\section{Introduction to {\it ab initio} molecular dynamics}
{\it Ab initio} molecular dynamics is used to computationally determine properties of matter, e.g.\ in material science and biomolecular systems, cf.\ \cite{marx-hutter}. The attractive part is the first principle approach, i.e.\ that empirical model parameters need not be set, and the drawback is the higher computational cost, as compared to molecular dynamics based on empirical
potentials. Nevertheless, the development of computer hardware and computational algorithms has stimulated a
growing activity in {\it ab initio} molecular dynamics simulations, see \cite{marx-hutter}.
In particular,  under certain assumptions,
{\it ab initio} molecular dynamics can be used to approximate the more fundamental quantum observables based on the time-independent Schr\"odinger equation, see \cite{lebris_acta}.
The {\it time-independent  Schr\"odinger} eigenvalue equation 
\begin{equation}\label{eq:schrod}
  \HOPER(X) \Phi(X) = E\Phi(X) \COMMA
\end{equation}
models many-body (nuclei-electron) quantum systems, 
where  the eigenvalue $E\in\rset$ is the energy of the system and  $\Phi:\rset^d\rightarrow \mathbb C^n$ is the corresponding 
eigenfunction for the Hamiltonian operator
\[ 
   \HOPER(X)=  - \frac{1}{2} M^{-1}I_n \Delta + \VOPER(X)\COMMA
\] 
with the nuclei position 
coordinates $X \in \rset^d,$  the $n\times n$ identity matrix $I_n$,  
the nuclei-electron mass ratio $M$, the Laplace operator $\Delta=\sum_{j=1}^d\frac{\partial^2}{\partial X_j^2}$
and the $n\times n$ Hermitian
matrix $\VOPER(X)$. For simplicity in the presentation we consider an $n\times n$ matrix approximation $V$  of the electron kinetic energy, electron-electron Coulomb repulsion potential and electron-nuclei Coulomb interaction potential. Since no unknown parameters are involved it makes the Schr\"odinger equation an {\it ab initio} model which forms
the basis for computational chemistry, cf. \cite{lebris_acta}.
In high dimension $d\gg 1$ equation \eqref{eq:schrod} is computationally expensive to solve.
%
The Born-Oppenheimer molecular dynamics method with 
constant number of particles, volume and total energy, based on \cite{BO} and formulated in 
e.g. \cite{lebris_hand,lebris_acta},  is 
less expensive and the basic
computational alternative 
for approximating quantum observables such as $\int_{\rset^d}g(X)|\Phi(X)|^2dX$  for desired functions $g:\rset^d\rightarrow\rset$, 
when $M\gg 1$, $d\gg 1$ and the smallest  and second smallest eigenvalues of the potential matrix $V(X)$ are separated, see \cite{how-accurate-BO}.
The {\it Born-Oppenheimer ab initio molecular dynamics} is given by the Hamiltonian system 
\[ 
\begin{split}
\dot X_\tau &= \nabla_P  H_{bo}(P_\tau,X_\tau) \\
\dot P_\tau & = -\nabla_X   H_{bo}(P_\tau,X_\tau) \COMMA
\end{split}
\] 
with the Hamiltonian  $ H_{bo}(P,X) := |P|^2/(2M)+ \lambda_0(X)$ where $X_\tau\in\rset^d$ and $P_\tau\in\rset^d$ are the nuclei 
position and momentum coordinates, respectively, at time $\tau$, and $\lambda_0:\rset^d\rightarrow \rset$ is the 
smallest 
eigenvalue of the electron eigenvalue problem
\[
\VOPER(X) \Psi_i(X) = \lambda_i(X) \Psi_i(X),\quad i = 0,1,2,\ldots,n-1\COMMA
\]
for fixed nuclei position $X$. The initial condition is chosen to satisfy $H_{bo}(P_0,X_0)=E$, so that the total energy of the system is $E$.
By changing to the slower timescale $t = M^{-1/2}\tau$ we obtain the Born-Oppenheimer Hamiltonian 
$H_{0}(P,X) := |P|^2/2 + \lambda_0(X)$ and the dynamics 
\begin{equation}\label{eq:BO-scaled}
\begin{split}
\dot X_t &= P_t\, ,\\
\dot P_t&= -\nabla\lambda_0(X_t)\, ,
\end{split}
\end{equation}
where the nuclei move a distance of order one in time one, independent of the large mass ratio parameter $M$.
Time discretization of~\eqref{eq:BO-scaled} requires that the ground state eigenvalue $\lambda_0$ is determined 
in each time step. For a large electron system, it can be computationally
expensive to solve the ground state electron eigenvalue problem
\begin{equation}\label{eq:BO-ground-state-potential}
\VOPER(X) \Psi_0(X) = \lambda_0(X) \Psi_0(X),
\end{equation}
although efficient iterative eigenvalue algorithms has improved the efficiency, see \cite{lebris_acta}.
Car-Parrinello molecular dynamics as formulated in the seminal work \cite{car_par}, see also \cite{marx-hutter,lebris_acta}, is an alternative to Born-Oppenheimer molecular dynamics,
where the solution of the algebraic eigenvalue problem is approximated by a relaxation method with  fictitious electron dynamics,
avoiding explicit solution of the eigenvalue problem~\eqref{eq:BO-ground-state-potential}.
{\it Ab initio} molecular dynamics, in fact, became more widely used after the introduction of the Car-Parrinello computational method, see ~\cite{marx-hutter}.  Computational comparisons  between Born-Oppenheimer and
Car-Parrinello dynamics  currently offer no general winner, see \cite{marx-hutter}, \cite{tangney}.
Ehrenfest molecular dynamics~\cite{marx-hutter}, which is closely related to the Car-Parrinello method, has the same property to determine the 
ground state dynamically without explicitly solving an electron eigenvalue problem in each time-step.
Theoretical evaluation of the Car-Parrinello method in \cite{Pastori, bornemann_schutte} shows the compromise between computational accuracy and work: accurate simulations require larger fictitious nuclei-electron  mass ratio as the eigenvalue gap, 
$\max_X(\lambda_1(X)-\lambda_0(X) )$, becomes smaller 
and the computational work, proportional to the number of time steps,  
increases as the fictitious mass ratio increases due to faster oscillations in the wave function. 
 A non trivial computational issue in the Car-Parrinello method is to determine the fictitious  electron mass relaxation parameter. In fact the lack of a precise method to determine the fictitious mass parameter is often mentioned as the main drawback of the Car-Parrinello method, cf. \cite{lebris_acta}.

Bornemann and Sch\"utte \cite{bornemann} presented an algorithm to automatically determine the fictitious electron mass parameter and to improve
computational efficiency by dynamically choosing the parameter as a time-dependent function which is piecewise 
constant on chosen time intervals. Their algorithm is based on limiting the maximum value of the fictitious electron kinetic energy and the result depends on the length of the time-intervals, which needs to be
optimized in numerical experiments; the time intervals are needed to average out the oscillatory behavior of the fictitious kinetic energy.
Inspired by \cite{bornemann}, 
the aim of this work is to construct and analyze an improved adaptive algorithm for determining the fictitious mass ratio with less parameters, which is also applicable to Ehrenfest dynamics, and  further improve the successful Car-Parrinello method. A second aim is to compare the computational efficiency for Ehrenfest and Car-Parrinello dynamics.

The {\it Car-Parrinello molecular dynamics}, in the settings of the simple example problems in this paper, is given by
\begin{equation}\label{eq:cp}
\begin{split}
\dot X_t & = P_t,\\
\dot P_t &= -\frac{\langle \psi_t,\nabla V(X_t)\psi_t\rangle}{\langle \psi_t,\psi_t\rangle},\\
\frac{1}{M_{CP}}\ddot \psi_t &=  -V(X_t)\psi_t + \Lambda\psi_t,
\end{split}
\end{equation}
where $M_{CP}$ is the fictitious nuclei-electron mass ratio, $\psi:[0,\infty)\rightarrow \mathbb{C}^n$ 
the electron wave function, 
and $\Lambda\in\rset$ the Lagrangian variable for the constraint $|\psi_t| = 1$. Here $\langle\cdot, \cdot\rangle$ denotes the standard scalar product in $\mathbb C^n$.
The fictitious nuclei-electron mass ratio $M_{CP}$ can be considered as a kind of control parameter, and by choosing $M_{CP}(t)$ adaptively, as a 
time-dependent function, the 
computational process can be made more precise and efficient.

The related {\it Ehrenfest molecular dynamics} is given by
\begin{equation}\label{eq:ehrenfest}
\begin{split}
\dot X_t & = P_t\\
\dot P_t &= -\frac{\langle \psi_t,\nabla V(X_t)\psi_t\rangle}{\langle \psi_t,\psi_t\rangle}\\
\frac{i}{M_E^{1/2}}\dot \psi_t &= V(X_t)\psi_t \COMMA
\end{split}
\end{equation}
where $M_E$ is the artificial nuclei-electron mass ratio parameter, $\psi_t\in \mathbb C^n$ the electron wave function and $i$ the imaginary unit. 

The Car-Parrinello method is often formulated for the electron ground state determined by the Hartree-Fock method or the Kohn-Sham density functional method, based on electron orbitals $\phi_\nu:\rset^3\rightarrow\rset, \ \nu=1,2,3,\ldots, n$. 
Let us quickly review the Hartree-Fock method  to compare it to the  formulations  \eqref{eq:ehrenfest}, \eqref{eq:cp} and \eqref{eq:BO-scaled}.
In a basis $\{\chi_j\}_{j=1}^R$, with $\chi_j:\rset^3\rightarrow\rset$, the {\it Hartree-Fock ground state} energy can be written
\[
E^{HF}=\mbox{trace}(hCC^*) + \frac{1}{2}\mbox{trace}(GCC^*)\, ,
\]
see \cite{lebris_acta}, which by the Euler-Lagrange equations lead to a non linear ground state eigenvalue problem.
Here the electron orbitals are assumed to be orthogonal, i.e. $C^*SC=I_n$, with $I_n$ denoting the $n\times n$ identity matrix,
 $C$ is the $R\times n$ coordinate matrix, i.e.  $\phi_\nu(x)=\sum_{j=1}^RC_{j\nu}\chi_j(x)$,  
$S$ is the $R\times R$ overlap matrix $S_{jk}:=\int_{\rset^3}\chi_j(x)\chi_k(x) dx$ and
the $R\times R$ Hermitian matrices $h$ and $G$ are defined as
\[
\begin{split}
h_{ij} &:= \frac{1}{2} \int_{\rset^3}\nabla \chi_i(x)\nabla\chi_j(x) dx
+\big(\sum_{k=1}^n \frac{z_k}{|x-X_k|} + \sum_{ k\ne\ell } \frac{z_k z_\ell }{|X_k-X_\ell |}\big) S_{ij}\\
D &:= CC^*\, ,\\
\Gamma_{jk}^{\ell m} 
 &:=  \int_{\rset^3} \int_{\rset^3} \frac{\chi_j(x)\chi_k(x)\chi_\ell(x')\chi_m(x')}{|x-x'|} dx dx'\\
J_{jk} &: = \sum_{\ell, m =1}^R \Gamma_{jk}^{\ell m} D_{\ell m}\, ,\\
K_{jk} &:= \sum_{\ell, m =1}^R \Gamma_{jm}^{k \ell}  D_{\ell m}\, ,\\
G_{jk}&:=  J_{jk} - K_{jk}\, .\\
\end{split}
\]
The Ehrenfest dynamics corresponding to the Hartree-Fock ground state is
\begin{equation}\label{TDHF}
\begin{split}
\ddot X_t &= - \mbox{trace} \big(\nabla_X h(X_t) C_t C_t^*\big)\, ,\\
\frac{i}{M_E^{1/2}} S\dot C_t &= F(C_tC_t^*)C_t\, , 
\end{split}
\end{equation}
with the Fock matrix $F(CC^*):= h+G$ and initially $C_0^*SC_0=I_n$. 
This system is also called the time-dependent Hartree-Fock method.
Note that the dynamics implies $C_t^*SC_t=I_n$ for all time $t$, since $F$ is Hermitian.
Equation \eqref{eq:ehrenfest} is related to  \eqref{TDHF} by using the scalar product
 $\mbox{trace}(AB^*)$ of two matrices $A,B$ and letting $\psi=C_{\cdot j}$:
 the time-dependent Hartree-Fock method \eqref{TDHF} can be written as
 \begin{equation}\label{TDHFn}
 \begin{split}
\ddot X_t &= - \sum_{j=1}^n \langle C(t)_{\cdot j},\nabla_X h(X_t) C(t)_{\cdot j}\rangle\, ,\\
\frac{i}{M_E^{1/2}} S\dot C(t)_{\cdot j} &= F(C_tC_t^*)C(t)_{\cdot j}\, ,\quad j=1,2,\ldots, n, 
\end{split}
\end{equation}
which shows that the Ehrenfest dynamics \eqref{eq:ehrenfest} has the form of a time-dependent Hartree-Fock equation with one orbital but without the non linear dependence on $CC^*$ in the Fock matrix $F$.
Similarly the Hartree-Fock Car-Parrinello method takes the form
\[
\begin{split}
\ddot X_t &= - \mbox{trace} \big(\nabla_X h(X_t) C_t C_t^*\big)\, ,\\
\frac{1}{M_{CP}} S\ddot C_t &= -F(C_tC_t^*)C_t + SC_t\Lambda_t\, , 
\end{split}
\]
where $\Lambda_t$ is the $n\times n$ matrix of Lagrange multipliers for the constraints $C_t^*SC_t=I_n$.

Section \ref{sec:adaptive_mass} formulates algorithms
to automatically choose
the artificial mass ratio parameter for  Car-Parrinello and Ehrenfest molecular dynamics. The algorithms  are motivated by the Landau-Zener probability, as explained in Section \ref{sec:LZ}, and include only one parameter, which is problem independent, in contrast to the problem dependent constant mass ratio parameter. 
The numerical experiments in Section \ref{sec:nr} show that the  choice of time-dependent adaptive mass ratio also
makes the computational methods more efficient, and  Car-Parrinello dynamics is slightly more efficient than Ehrenfest dynamics. Section \ref{sec:alt} presents an alternative adaptive mass algorithm and  the appendix derives an estimate of the probability to be in excited states for Ehrenfest dynamics, which also motivates the adaptive choice of the mass ratio for a general Ehrenfest dynamics model.
\section{Adaptive mass algorithms}\label{sec:adaptive_mass}
In this section we present adaptive mass ratio algorithms  for Ehrenfest  and Car-Parrinello  dynamics, in \eqref{eq:ehrenfest}  and \eqref{eq:cp},
approximating the ground state Born-Oppenheimer
molecular dynamics in \eqref{eq:BO-scaled}.

Time discretizations of the Car-Parrinello differential equation \eqref{eq:cp} require solution of a nonlinear system of equations,  related to the orthogonality constraints,  in each time step.
Ehrenfest dynamics automatically satisfies these constraints which yields simpler computations.
Section \ref{sec:nr} shows on the other hand that Ehrenfest dynamics requires 
somewhat larger number of time-steps than the Car-Parrinello method, although the difference
is smaller for the adaptive mass algorithms, see~Figure~\ref{fig:1D-delta-pE}.

Ehrenfest molecular dynamics can effectively provide the electron ground state 
for a molecular system, provided the electron ground state eigenvalue problem is solved initially.
For a case when the electron potential surfaces have an avoided crossing, i.e.~the spectral
gap between the smallest electron eigenvalues is small, a large nuclei-electron mass ratio $M_E$ 
is needed to approximate Born-Oppenheimer dynamics by Ehrenfest dynamics, see~\cite{how-accurate-BO}.
On the other hand the computational work (proportional to the number of time steps)
 to simulate Ehrenfest molecular dynamics is of order $\mathcal{O}(\sqrt{M_E})$,
since the length of the time-steps can be at most $\mathcal{O}(M_E^{-1/2})$
to resolve  
the wave function oscillations with frequency proportional to 
$\sqrt M_E$. 

Hence, using a large mass ratio $M_E$  is computationally
expensive. However, for many problems the difference of the smallest eigenvalues $\lambda_1(X)-\lambda_0(X)$
is small only  in a small $X$-region. In such cases
it is computationally beneficial to use a large mass ratio close to the avoided crossing 
and to use  smaller mass ratios in the other parts of the domain.
Here we formulate an adaptive algorithm to dynamically determine a fictitious mass ratio  for Ehrenfest
and Car-Parrinello molecular dynamics, particularly suited for molecular dynamics applications with
nearly crossing potential surfaces.

We choose the {\it adaptive mass ratio functions} $M_E:\rset_+\rightarrow\rset_+$ and 
$M_{CP}:\rset_+\rightarrow\rset_+$ for Ehrenfest and Car-Parrinello molecular dynamics
based on the following criteria:
\begin{equation}\label{eq:adaptiveM-EH}
M_E(t) := \epsilon^{-2}\max\left(1,\big(\frac{|P_t|^{1/2}}{ |\mu_1(X_t)-\mu_0(X_t)|}\big)^{\gamma_E}\right)\COMMA
\end{equation}
and
\begin{equation}\label{eq:adaptiveM-CP}
M_{CP}(t) := \epsilon^{-2}\max\left(1,\big(\frac{|P_t|^{1/2}}{|\mu_1(X_t)-\mu_0(X_t)|}\big)^{\gamma_{CP}}\right)\COMMA
\end{equation}
where $\gamma_E$ and $\gamma_{CP}$ are positive constants optimally chosen by 
$\gamma_E=4$ and $\gamma_{CP}=2$, as explained in Section \ref{sec:LZ}. 
The constant (problem independent) error tolerance parameter $\epsilon$ is related to the probability to be in the excited states, see~\eqref{eq:epsilon-delta},
and the electron eigenvalues are approximated by the Rayleigh quotients
\[
\begin{split}
\mu_0(X_t) & := \frac{\langle\psi_t,V(X_t)\psi_t\rangle}{\langle\psi_t,\psi_t\rangle} \COMMA \\ 
\mu_1(X_t) & := \frac{\langle\dot\psi_t, V(X_t)\dot\psi_t\rangle}{\langle\dot\psi_t, \dot\psi_t\rangle}.
\end{split}
\]
Since the difference between the smallest electron eigenvalues  
can be arbitrarily large, there is a lower bound $1/\epsilon^2$ for the artificial mass ratio functions to avoid too large time steps.


\section{Adaptive masses motivated by transition probabilities}\label{sec:LZ}
This section motivates the criteria for the adaptive mass ratio algorithms \eqref{eq:adaptiveM-EH} and \eqref{eq:adaptiveM-CP} from the Landau-Zener probability.
The Landau-Zener model~\cite{Z} was introduced to model transitions to an excited state and
it is given by
\[ 
iM_{LZ}^{-1/2}\dot\phi_t= \left[\begin{array}{cc}
P_*t & \delta\\
\delta & -P_*t\\
\end{array}\right]
\phi_t
\]
where $\phi:\rset\rightarrow \mathbb C^2$ is the wave function, 
$i$ is the imaginary unit, $t$ is the time parameter,
$M_{LZ}$, $P_*$, and $\delta$ are constant positive parameters,
and $\lim_{t\rightarrow-\infty}\phi(t)= (1,0)$ is the initial data.
In the Landau-Zener model the matrix has the eigenvalues $\lambda_\pm(P_*t)=\pm \sqrt{(P_*t)^2+\delta^2}$
which shows that the electron potential surfaces $\lambda_\pm$ cross at $t=0$  if $\delta=0$ and there is an 
avoided crossing if $\delta$ is positive.
The transition probability in the Landau-Zener model is given
by
\begin{equation}\label{eq:p_LZ}
p_{LZ}:=\lim_{t\rightarrow\infty}|\phi_2(t)|^2 =e^{-\pi\delta^2 M_{LZ}^{1/2}/P_*}\COMMA
\end{equation}
which is the so-called {\it Landau-Zener probability}.

The Ehrenfest molecular dynamics \eqref{eq:ehrenfest}
and the Landau-Zener model are closely related and  the wave functions  are equal, $\psi=\phi$, in the case 
$X_t=P_*t$, $M_E=M_{LZ}$ and 
\[
V(X)=
\left[\begin{array}{cc}
X_t & \delta\\
\delta & -X_t\\
\end{array}\right]\, .
\]
 Since we are interested in approximating  Born-Oppenheimer molecular dynamics, which is the molecular dynamics with the electrons in their ground state, minimizing the transition probability
$p_{LZ}$ is of key interest. If the transition probability is small the Ehrenfest molecular dynamics
solution will be close to the Born-Oppenheimer molecular dynamics solution.

Rewriting equation \eqref{eq:p_LZ} for the Landau-Zener transition probability $p_{LZ}$ we obtain
\[ 
M_{LZ} = \frac{P_*^2\log^2p_{LZ}}{\pi^2\delta^4}.
\] 
Since $M_E \sim M_{LZ}$ we can relate this to \eqref{eq:adaptiveM-EH} as follows
\begin{equation}\label{eq:epsilon-delta}
\epsilon \sim \frac{\pi}{|\log p_{LZ}|}\COMMA \qquad  |\mu_1(X_t)-\mu_0(X_t)| \sim \delta\COMMA\qquad 
|P_t|=P_*\COMMA\qquad
\gamma_E=4\COMMA
\end{equation}
The chosen constant $\epsilon$, which relates  to the probability to be in  excited states, is problem independent
and $|\mu_1(X_t)-\mu_0(X_t)|/|P_t|^{1/2}$ measures 
the spectral gap locally, which is problem dependent. Experimental tests of $\gamma_E$ in Figure 
\ref{fig:EH-variableExpn} confirm the optimum $\gamma_E=4$.

In  the simplest case of a  constant potential $V$ the second order differential operator in Car-Parrinello dynamics
can be factorized 
\[
\frac{d^2}{dt^2}+M_{CP}(V-\Lambda)=\big(\frac{d}{dt}+i\sqrt{M_{CP}(V-\Lambda)}\big)
\big(\frac{d}{dt}-i\sqrt{M_{CP}(V-\Lambda)}\big)
\]
into a backward and a forward Ehrenfest dynamics operator for the potential $\sqrt{V-\Lambda}$. This observation motivates to choose $M_{CP}$ equal to $M_E$ for $\mu_1-\mu_0$ replaced by $\sqrt{\mu_1-\mu_0}$, 
cf. Remark \ref{rem1}, which explains $\gamma_{CP}=2$ in \eqref{eq:adaptiveM-CP} as tested numerically in Figure \ref{fig:CP-variableExpn}.

A related way to motivate the relation $M_E\sim |\lambda_1(X_t)-\lambda_0(X_t)|^{-4}|P_t|^2$ and estimate the probability to be in exited states for Ehrenfest dynamics is to
use the orthogonal decomposition 
\[
\psi_t=\bar\psi_0(t)\oplus\psi_t^\perp=:\langle \psi_t,\Psi_0(X_t)\rangle\Psi_0(X_t)\oplus \psi_t^\perp,\]
where $\psi_t^\perp$ is in the orthogonal complement of $\Psi_0(X_t)$. Then
we have by Lemma 5.8 in \cite{md_langevin} the bound
$|\psi_t^\perp|=\mathcal O(M^{-1/2})$; a proof of this and the more precise  estimate $|\psi_t^\perp|=\mathcal O(M^{-1/2}\delta_*^{-2})$, as $M\rightarrow\infty$ with $\delta_*:=\min_{t}(|\lambda_1(X_t)-\lambda_0(X_t)|/|P_t|^{1/2})$, is in the appendix.
Hence the probability to be in excited states is bounded by
\[
p_E=\langle \psi_t^\perp,\psi_t^\perp\rangle=\mathcal O(M^{-1}\delta_*^{-4})\, .
\]
We conclude that the wave function in Ehrenfest dynamics approximates the ground state eigenvector.
Consequently the difference in the forces 
\[
\langle \psi, \nabla(V-\lambda_0)\psi\rangle/\langle\psi,\psi\rangle
=2\Re\langle \psi^\perp,\nabla(V-\lambda_0)\bar\psi\rangle 
+ \langle \psi^\perp, \nabla(V-\lambda_0)\psi^\perp\rangle/\langle\psi,\psi\rangle
\]
also vanishes  as $M\rightarrow\infty$ and the paths  remain close for some time. 

If the probability to be in excited states grows one can reinitialize the wave function $\psi$  by computing
a new ground state eigenvector; a growing probability to be in excited states increases the computed approximation of the Born-Oppenheimer Hamiltonian  $\bar H_0:=|P|^2/2 + \mu_0(X)$ so that $\bar H_0-E$
increased above a tolerance level  is a  simple criterion for reinitialization.

\begin{remark}[Adding scalars to $V$ does not matter]
\label{rem1}If $\psi$ solves Ehrenfest dynamics \eqref{eq:ehrenfest}, 
the function $\psi(t)e^{iM^{1/2}\int_0^t\lambda_0(X_s)ds}$ solves Ehrenfest dynamics \eqref{ehren_trans}
with $V$ replaced by $V-\lambda_0$. Therefore, adding a scalar to the potential only changes the wave function 
$\psi$ by a phase factor. Similarly in Car-Parrinello dynamics \eqref{eq:cp} the Lagrange multiplier satisfies
\[
\begin{split}
\Lambda 
&= \langle V\psi,\psi\rangle + M_{CP}^{-1}\langle \ddot\psi,\psi\rangle\\
&= \langle V\psi,\psi\rangle - M_{CP}^{-1}\langle \dot\psi,\dot\psi\rangle\\
\end{split}
\]
using the constraint $\langle \psi_t,\psi_t\rangle=1$. Therefore, adding a scalar to $V$ also adds the same scalar to $\Lambda$. 
\end{remark}



\section{Numerical results}\label{sec:nr}
This section tests for two simple model problems  the adaptive mass ratio algorithms \eqref{eq:adaptiveM-EH} and
\eqref{eq:adaptiveM-CP}.
\subsection{A two dimensional problem}\label{sec:2d-prob}
We  consider the two dimensional, time-independent Schr\"odinger equation~\eqref{eq:schrod} 
with the heavy-particle coordinate~$X = (X_1,X_2) \in \mathbb R^2$, the two-state 
light-particle coordinate $x \in \{x_-,x_+\}$, and the potential matrix
%
\begin{equation}\label{eq:2d-potential}
\VOPER(X) := \Big(\underbrace{\frac{1}{2}(X_1^2+\alpha X_2^2) + \beta\sin(X_1X_2)}_{=: \lambda_s(X)}\Big)I + \eta\left[\begin{array}{cc}
                               \tan^{-1}\frac{X_1-a_1}{\eta} & \tan^{-1}\frac{X_2-a_2}{\eta}\\
                               \tan^{-1}\frac{X_2-a_2}{\eta} & -\tan^{-1}\frac{X_1-a_1}{\eta}
                        \end{array}\right],
\end{equation}
where $I$ is the $2 \times 2$ identity matrix, $\alpha=\sqrt 2$, $\beta = 2$, $\eta = 1/2$, and
$a = (a_1,a_2)\in\mathbb{R}^2$ is a point in the two dimensional space.
For each $X$, we have the eigenvalue problem, $\VOPER(X)\Psi_\pm(X) = \lambda_\pm(X)\Psi_\pm(X)$  
with the eigenvalues \[\lambda_\pm(X) = \lambda_s(X) \pm \eta\sqrt{\left(\tan^{-1}\frac{X_1-a_1}{\eta}\right)^2 + \left(\tan^{-1}\frac{X_2-a_2}{\eta}\right)^2}\, .\]
We denote $\Psi_-(X)$ the ground state eigenvector, and $\Psi_+(X)$  the excited state eigenvector.
We choose the energy $E=2$ and for this case with $\beta=2$ the dynamics appears to be ergodic, see \cite{how-accurate-BO}. In the model there is a conical intersection between the 
two eigenvalue surfaces at $(a_1,a_2)$. Depending on the location of the conical intersection 
the crossing between the potential surfaces, $\lambda_\pm$, may either be inside the {\it classically allowed} 
region~$\{X: \lambda_-(X) \le E\},$ or in the  {\it classically forbidden} region~$\{X: \lambda_-(X) > E\}$
in which case no classical path can pass through the conical intersection.
In this example we choose the conical intersection to be near the boundary, but outside, of the classically allowed
region, see Figure~\ref{fig:eig-levels}. This position of the conical intersection
can be considered as a near avoided conical intersection
inside the classically allowed region, with smaller spectral gap the closer it is to the boundary.

We use the symplectic St\"ormer-Verlet discretization method, see \cite{verlet-method}, for the modified translated Ehrenfest molecular dynamics given by
\begin{equation}\label{eq:mod-ehrenfest}
\begin{split}
\dot X_t & = P_t\\
\dot P_t &= -\frac{\langle \psi_t,\nabla V(X_t)\psi_t\rangle}{\langle \psi_t,\psi_t\rangle}\\
iM_E^{-1/2}\dot \psi_t &= \left(V(X_t)-I\frac{\langle\psi_t,V(X_t)\psi_t\rangle}{\langle\psi_t,\psi_t\rangle}\right)\psi_t.
\end{split}
\end{equation}
Note that if $(X_t,P_t,\psi_t)$ solves the standard Ehrenfest dynamics \eqref{eq:ehrenfest}, then \[t\mapsto (X_t,P_t,\psi_te^{iM^{1/2}\int_0^t\frac{\langle\psi_s,V(X_s)\psi_s\rangle}{\langle\psi_s,\psi_s\rangle} ds} )\] solves the modified equation \eqref{eq:mod-ehrenfest}, with the advantage of having less oscillations in $\psi$  if it is close to the ground state.
We choose the initial data $X_0=(-2.0, 0.5)$, 
$P_0 = (1, \sqrt{2(E-\lambda_-)-1)})$, and $\psi_0$ to be the ground state eigenvector of $V(X_0)$.

\begin{figure}[h!]
\centering
\includegraphics[width=0.45\textwidth]{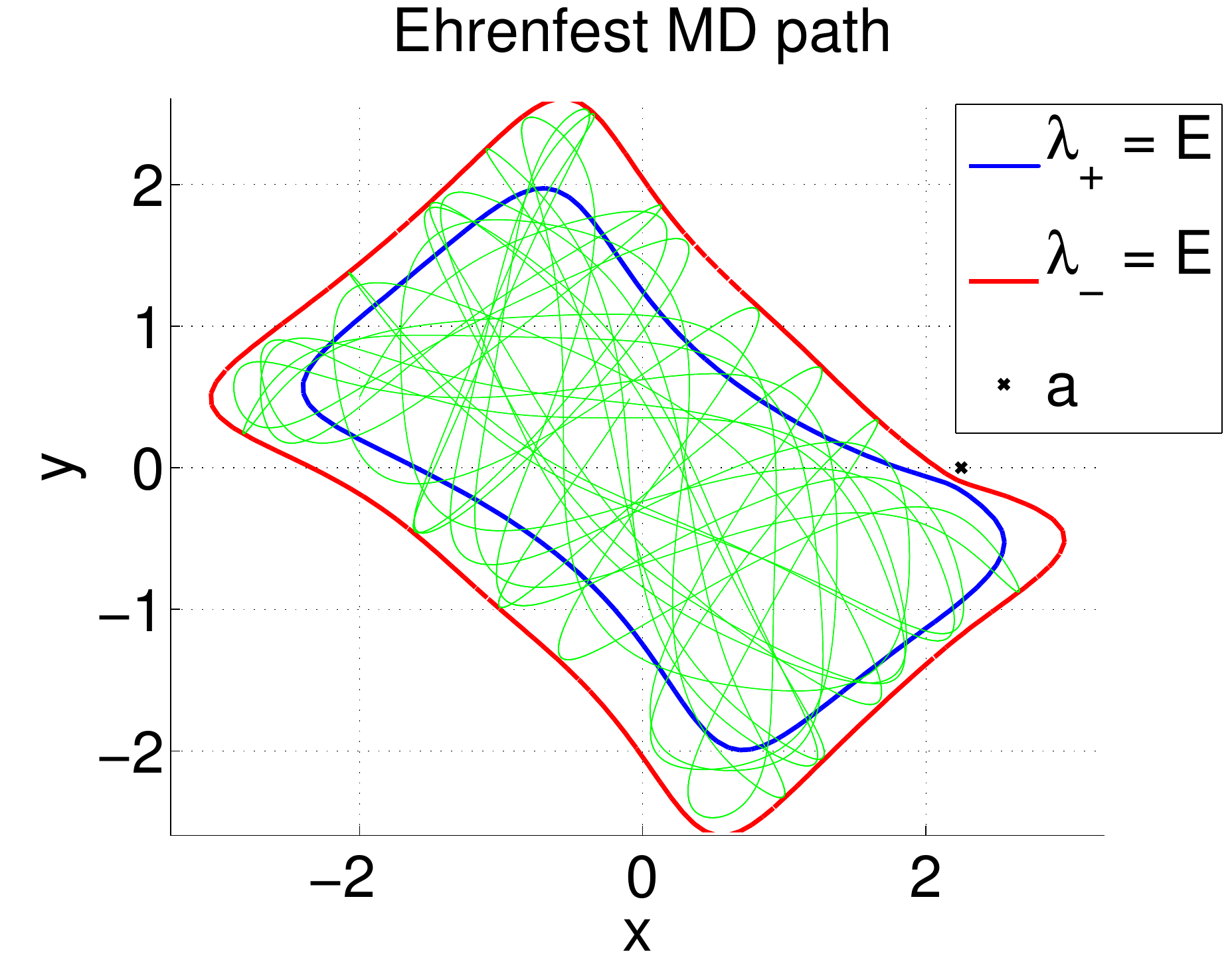}
\includegraphics[width=0.45\textwidth]{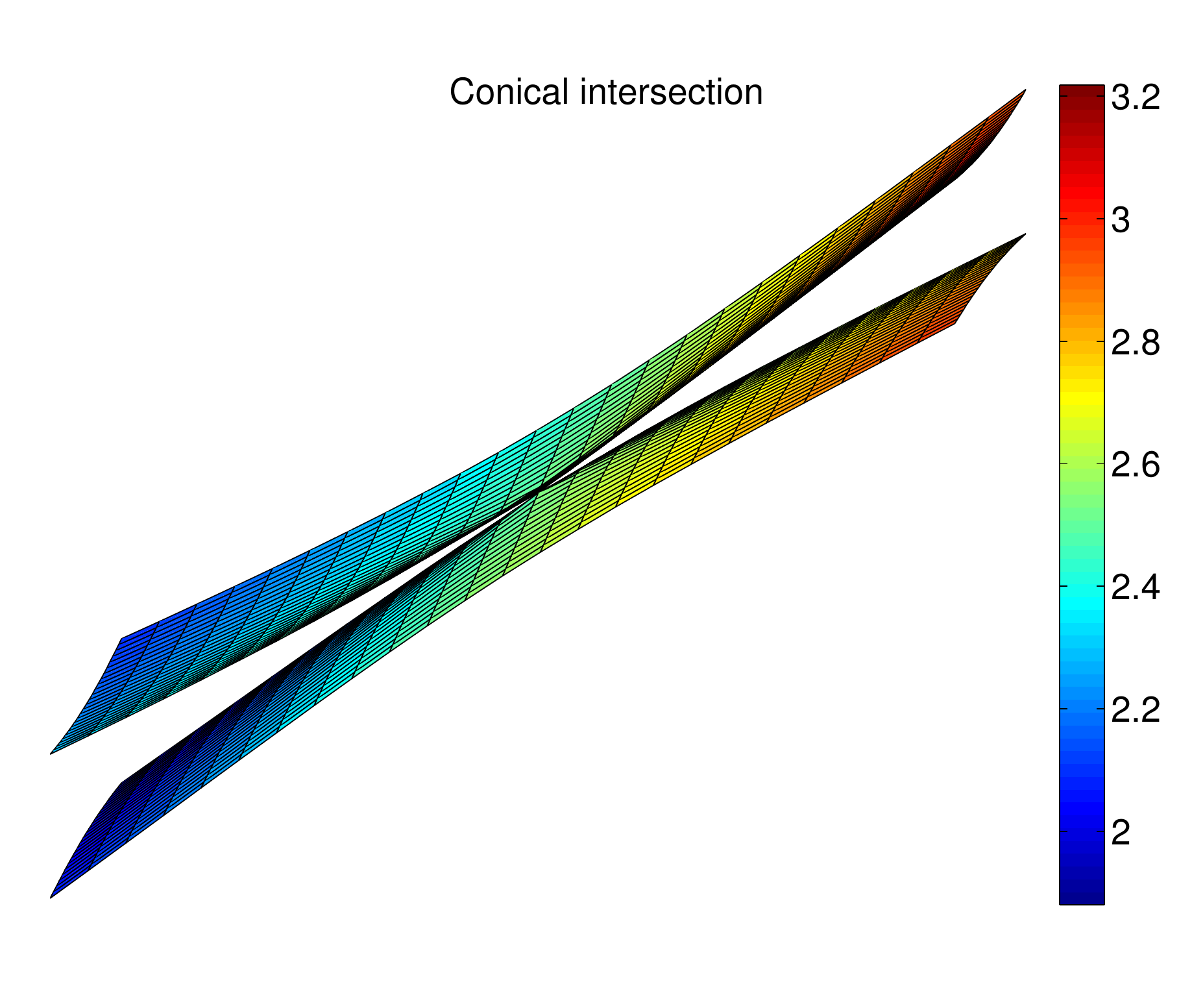}
\caption{(Left) Plots showing the molecular dynamics path and eigenvalue level curves at energy $E=2$. The area inside  the eigenvalue level
curve for $\lambda_0=E$ defines the classically allowed region. A small cross sign shows 
the conical intersection point $a = (2.25,0),$ which is near the boundary, but outside, of the classically allowed region. 
(Right) Potential surfaces near the conical intersection. The bottom and top surfaces
correspond to the ground state and excited state potential surfaces, respectively. 
The colors are mapped with the numerical values of the eigenvalues. 
}\label{fig:eig-levels}
\end{figure}

\begin{figure}[h!]
\centering
\includegraphics[width=0.45\textwidth]{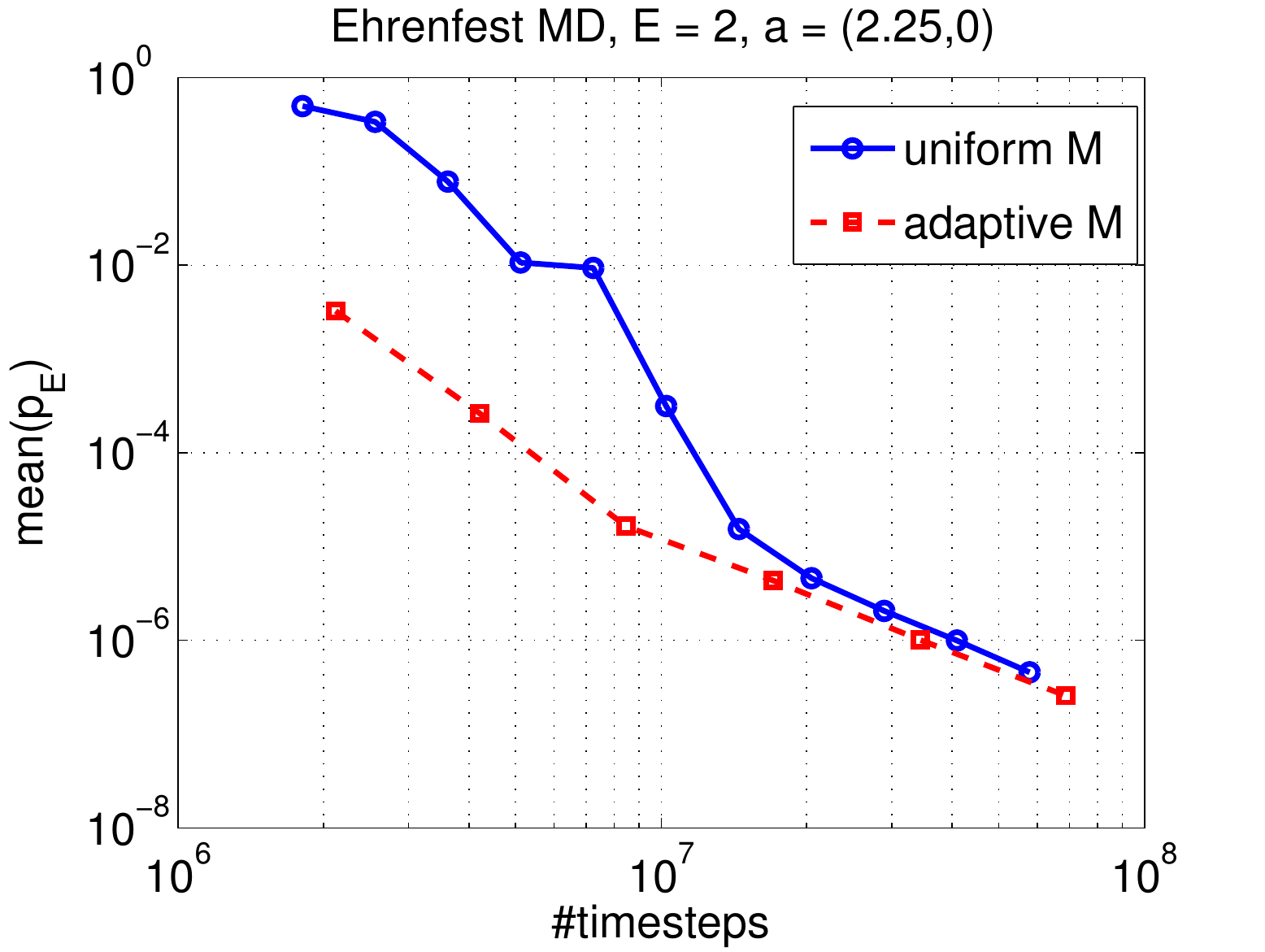} \\
\includegraphics[width=0.45\textwidth]{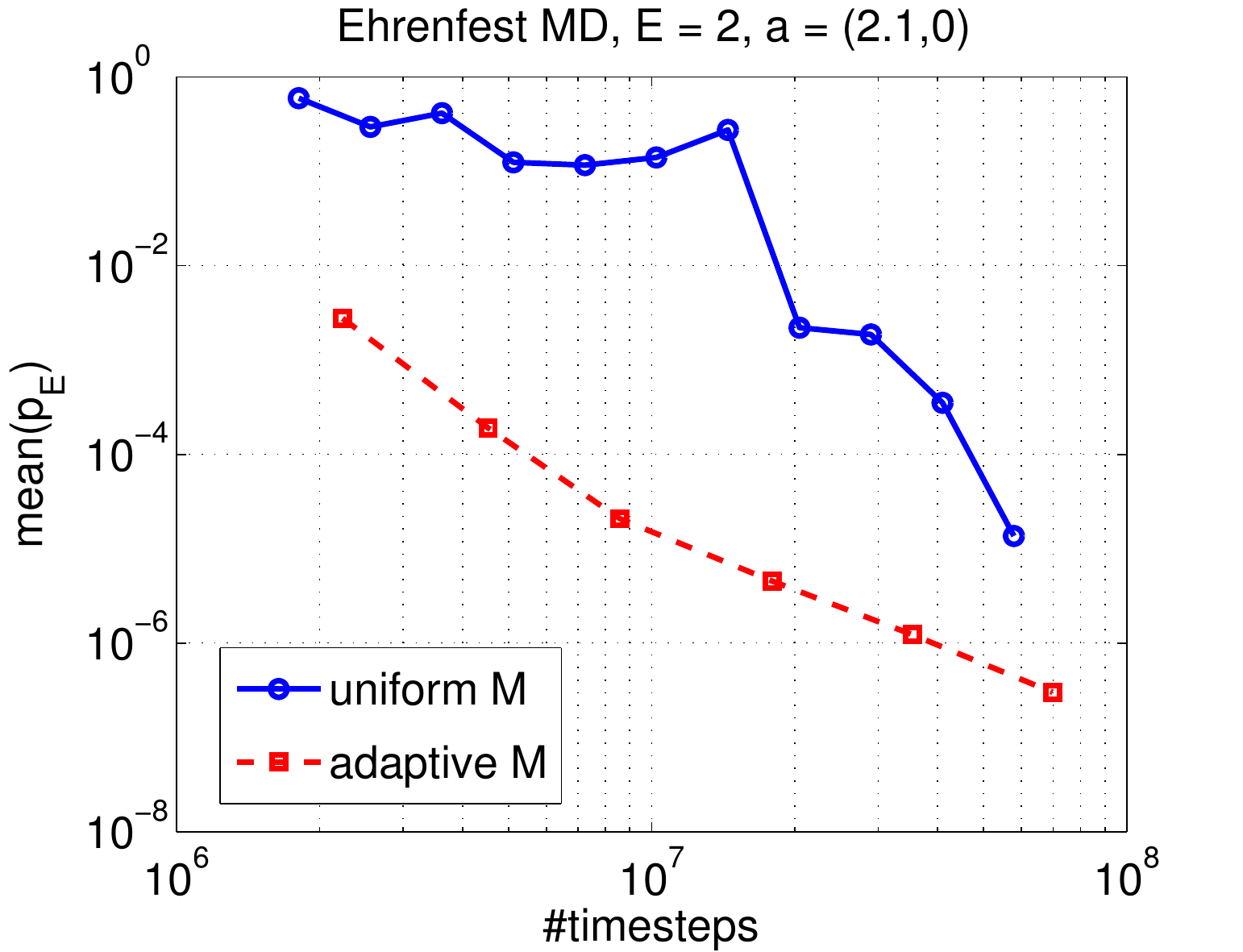} \\
\includegraphics[width=0.45\textwidth]{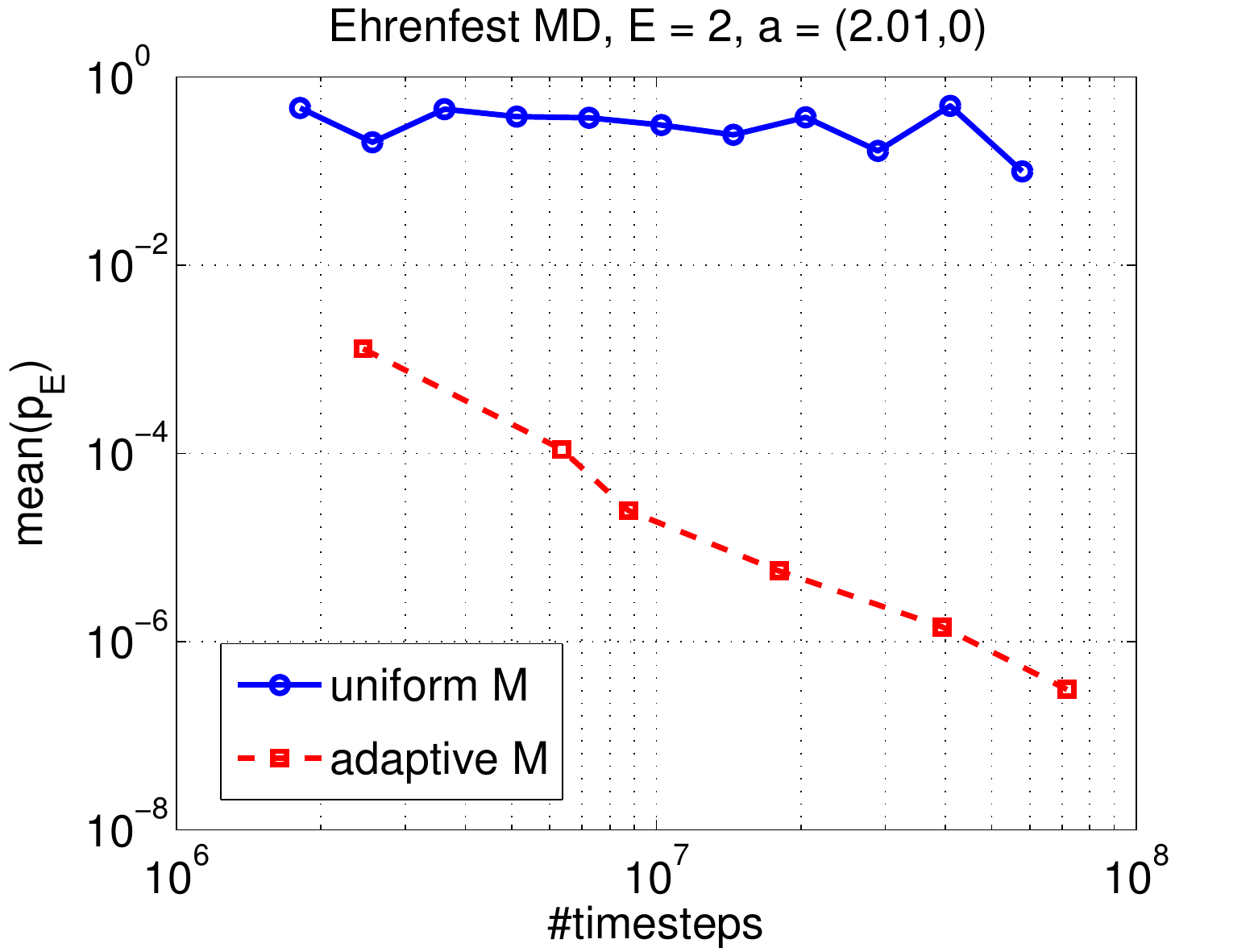}
\caption{Plots showing the arithmetic mean of the excitation probabilities, $p_E$, as a function of computational work.
As the spectral gap decreases from the top figure to the bottom, the computational work increases fast with uniform mass and slowly with adaptive mass.
The simulations are computed with time $t\in [0,2000]$,  for the 2D-problem \eqref{eq:2d-potential}.}\label{fig:mean-pE}
\end{figure}

\begin{figure}[h!]
\centering
\includegraphics[width=0.45\textwidth]{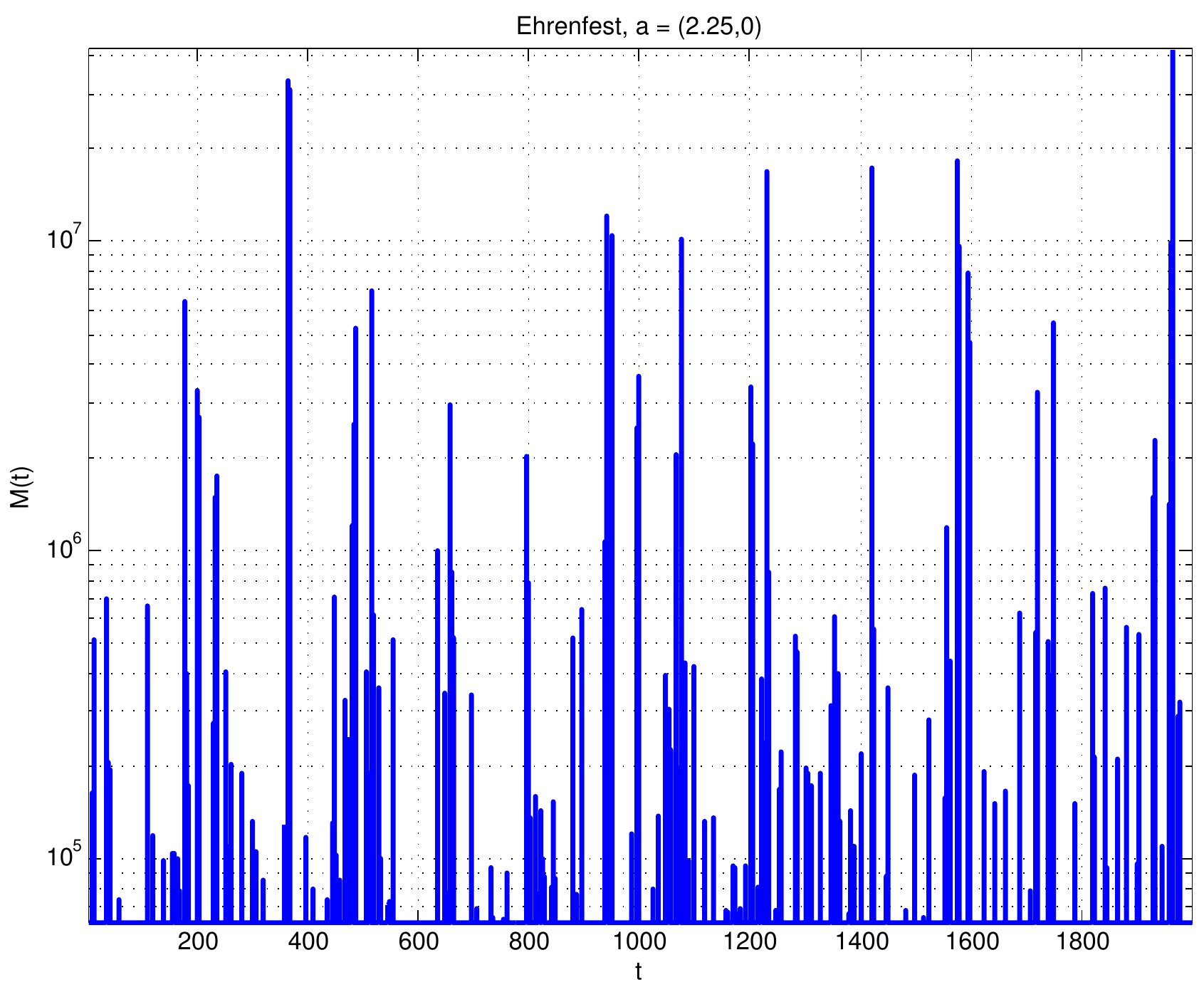}
\caption{Adaptive mass ratio distribution $M_E$ for Ehrenfest molecular dynamics ($\epsilon = 0.004$) 
as a function of time for the 2D-problem \eqref{eq:2d-potential}.}\label{fig:M-adaptive}
\end{figure}

The excitation probability, $p_E$, is defined by
\[
p_E(t) = \frac{|\langle \psi_t, \Psi_+(X_t) \rangle|^2}{\langle \psi_t, \psi_t \rangle \langle \Psi_+(X_t), \Psi_+(X_t) \rangle} \, .
\]

Figure~\ref{fig:mean-pE} shows the main results related to the idea of using adaptive mass in Ehrenfest
molecular dynamics simulations:
%
the figure compares the computational work required 
for different excitation probabilities $p_E$. Here, we are interested in solving a ground state problem, and in the ideal case $p_E$ 
is  small. A larger value of $p_E$ will indicate larger probability that the molecular
system shifts from the ground state to an excited state. 
As the spectral gap becomes smaller the number of time steps with uniform mass increases fast while the number of time steps with adaptive mass ratio increases only slowly.
When the computational work is  large the arithmetic mean of $p_E$ comes rather close for the adaptive and uniform mass cases. This computational result indicates that for sufficiently large mass
both uniform and adaptive mass based methods give similar mean value for $p_E,$ and, for a smaller 
mass, adaptive mass works better. In the case when the computational cost is an important
factor, the adaptive mass is therefore an attractive alternative since the number of time-steps for a certain
accuracy is smaller, especially for small spectral gaps. Figure~\ref{fig:M-adaptive} shows an example of the the distributions 
of adaptive mass, $M_E(t)$, as a function of time, $t$, where the size of the mass ratio varies many orders of magnitude.

The value of the Born-Oppenheimer molecular dynamics observable, $g_{MD}$, in the ergodic case is given by 
\begin{equation*}
  g_{MD} := \lim_{\delta \rightarrow 0+} \frac{\int_{E<H_0(X,P)<E+\delta} g(X,P)dXdP}{\int_{E<H_0(X,P)<E+\delta} dXdP}\COMMA
\end{equation*}
cf. \cite{reed_simon}, which for the two dimensional case, is given by 
\[
  g_{MD} = \frac{\int_{H_0(X,0)\le E} g(X)dX}{\int_{H_0(X,0)\le E} dX}\COMMA
\]
when $g$ depends only on $X$. We compute a reference value for $g_{MD}$ using Monte-Carlo integration
and compare it to Ehrenfest dynamics approximations.
The Ehrenfest dynamics observables
are approximated by the time average $\int_0^T g(X_t)\frac{dt}{T}$, where $X_t$ is the Ehrenfest dynamics position path.
Figure~\ref{fig:obs-adaptive} shows $\mathcal O(T^{-1/2})$ convergence rate of the Ehrenfest dynamics 
observables towards $g_{MD}$ for the adaptive mass cases applied to the two dimensional problem~\ref{sec:2d-prob}.

\begin{figure}[h!]
\centering
\includegraphics[width=0.45\textwidth]{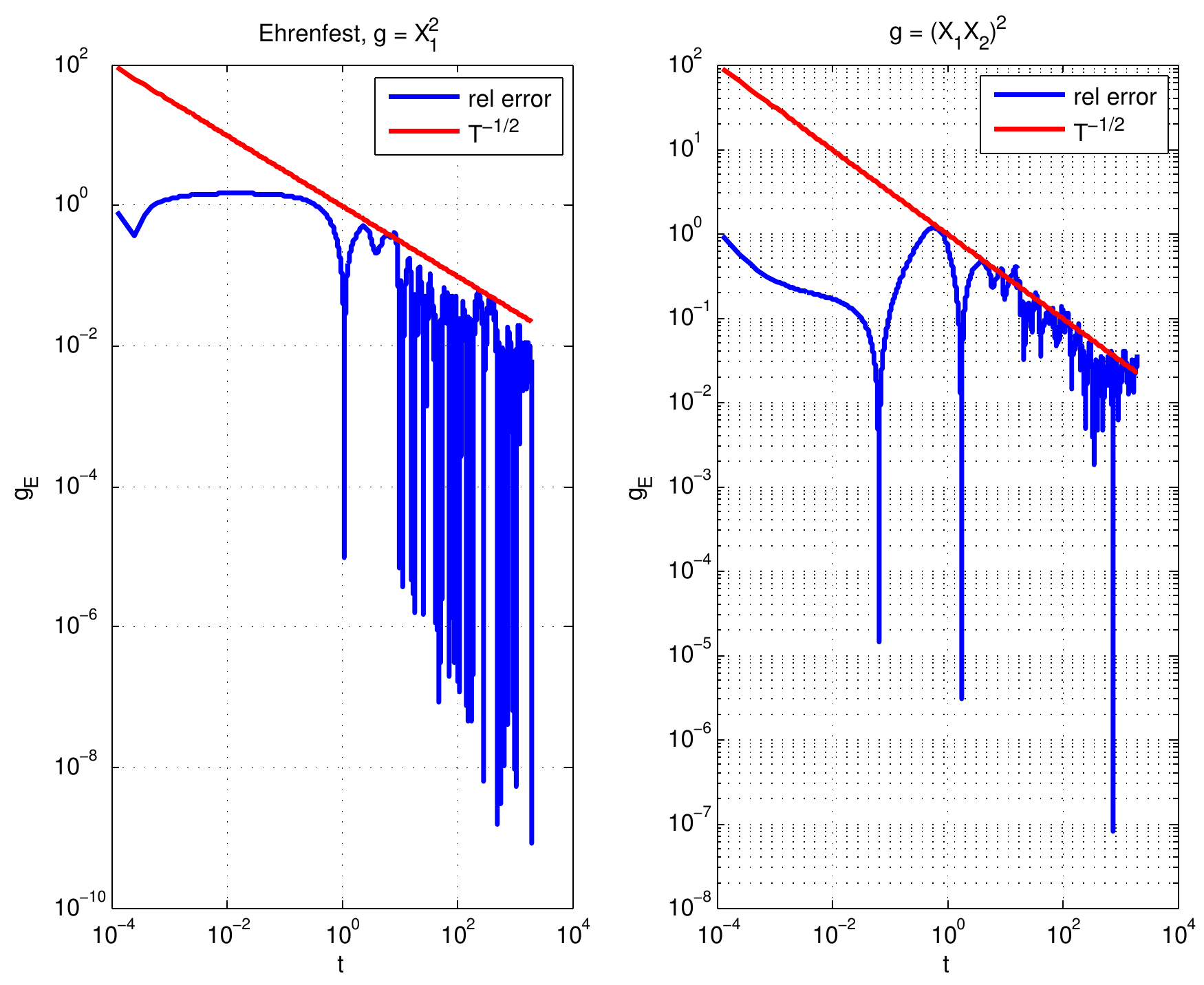}
\caption{Convergence of Ehrenfest molecular dynamics observable to 
molecular dynamics observables for adaptive mass case ($\epsilon = 0.004$) and  the 2D-problem \eqref{eq:2d-potential}.}\label{fig:obs-adaptive}
\end{figure}

\begin{figure}[h!]
\centering
\includegraphics[width=0.45\textwidth]{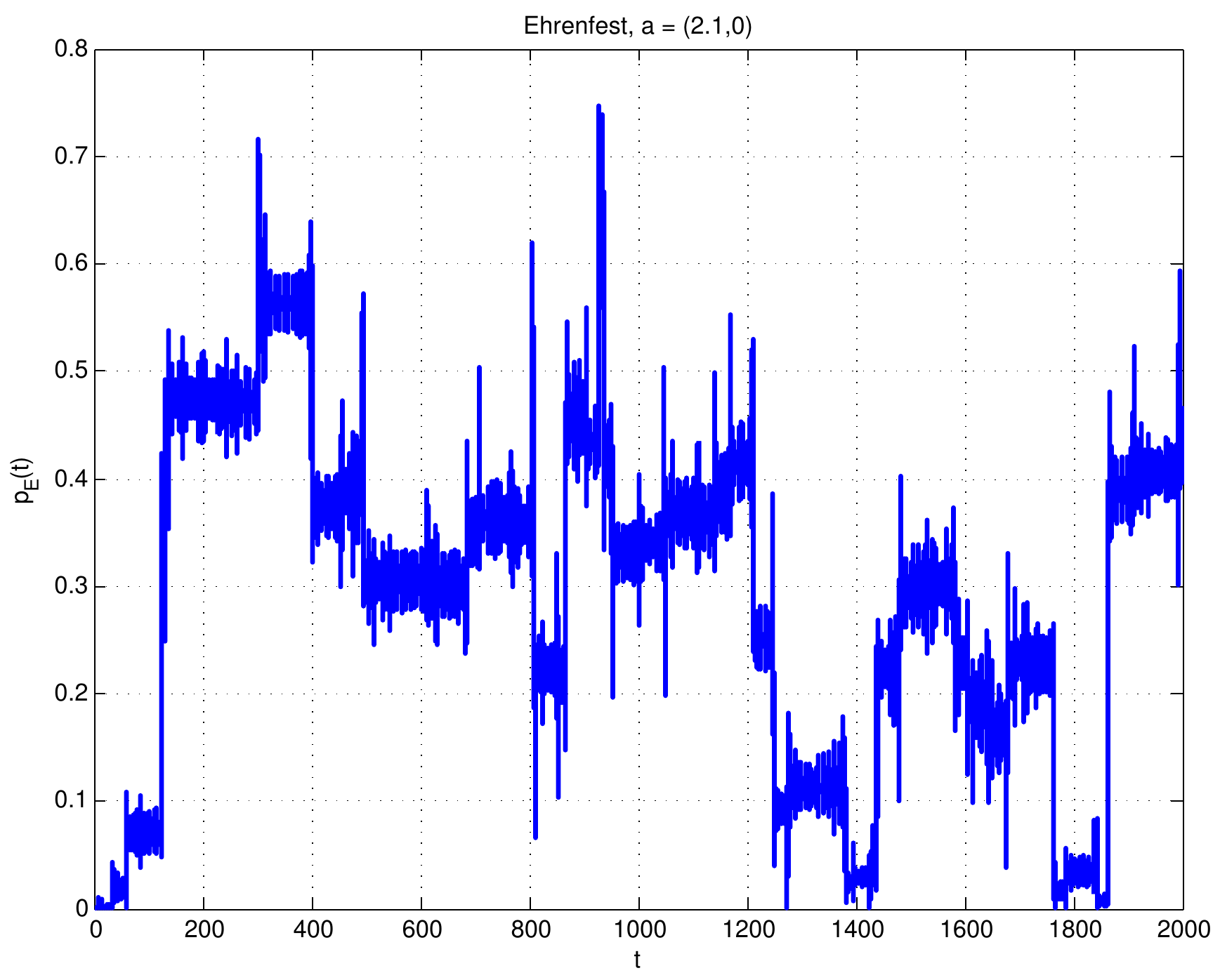}
\includegraphics[width=0.45\textwidth]{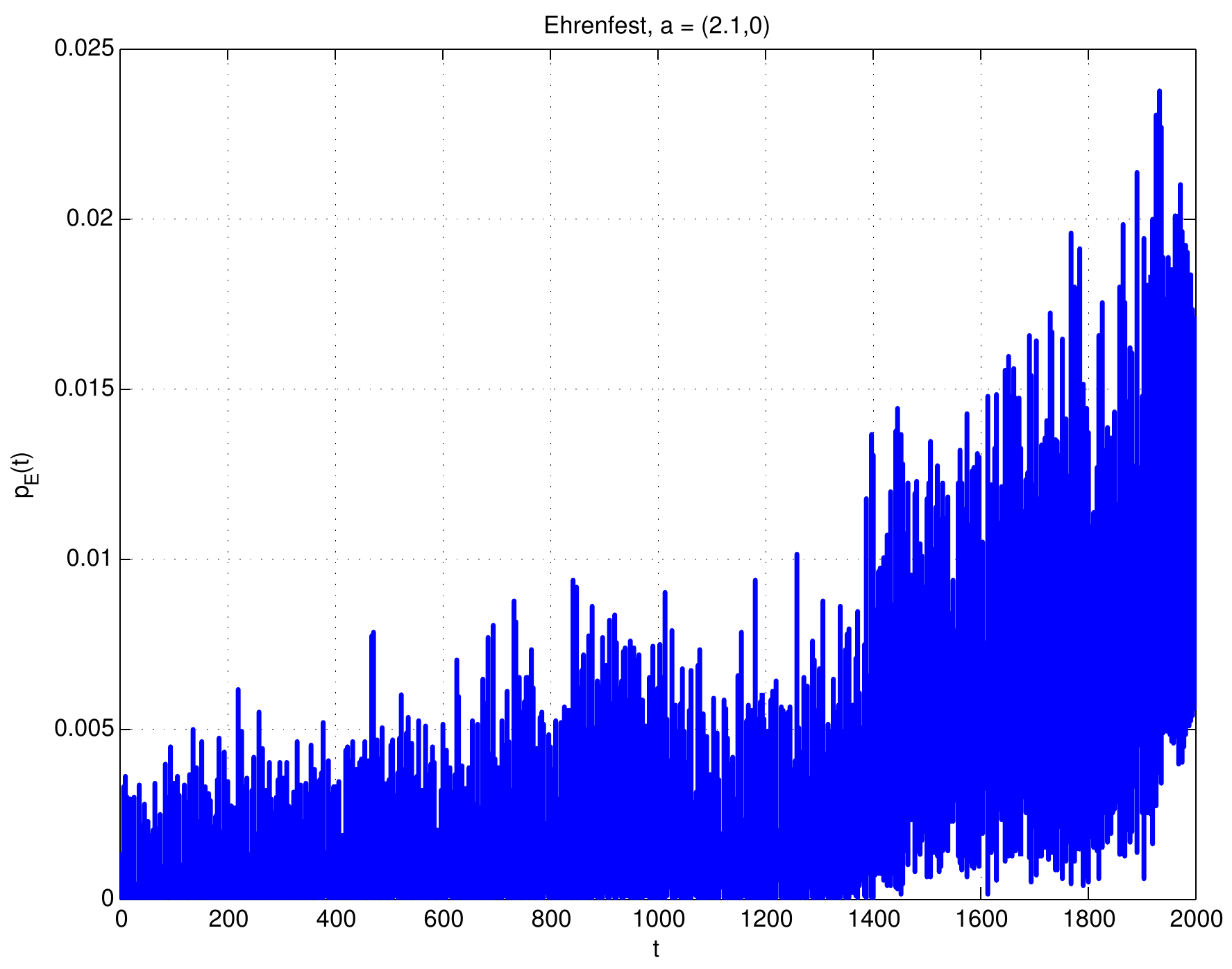} \\
\includegraphics[width=0.45\textwidth]{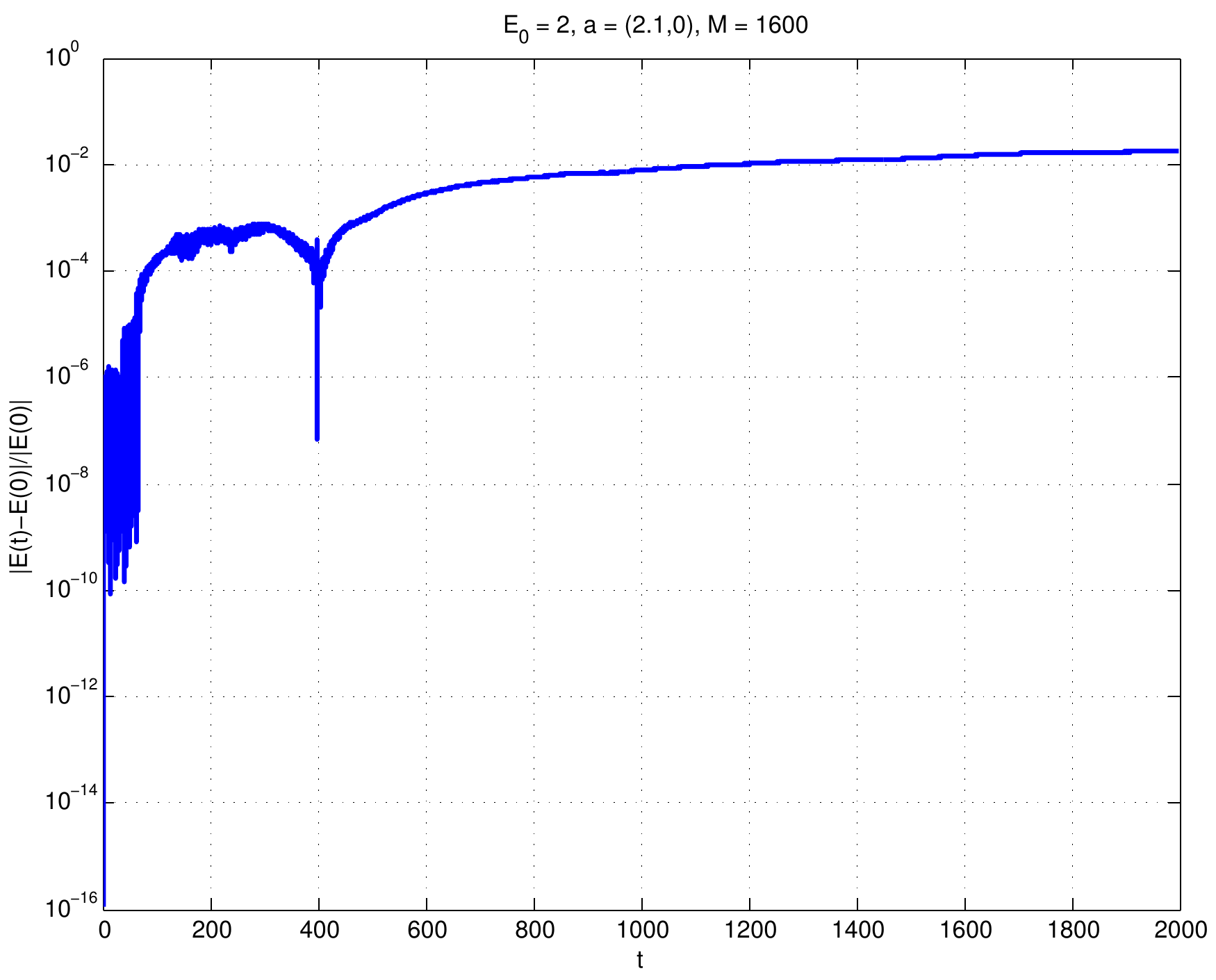}
\includegraphics[width=0.45\textwidth]{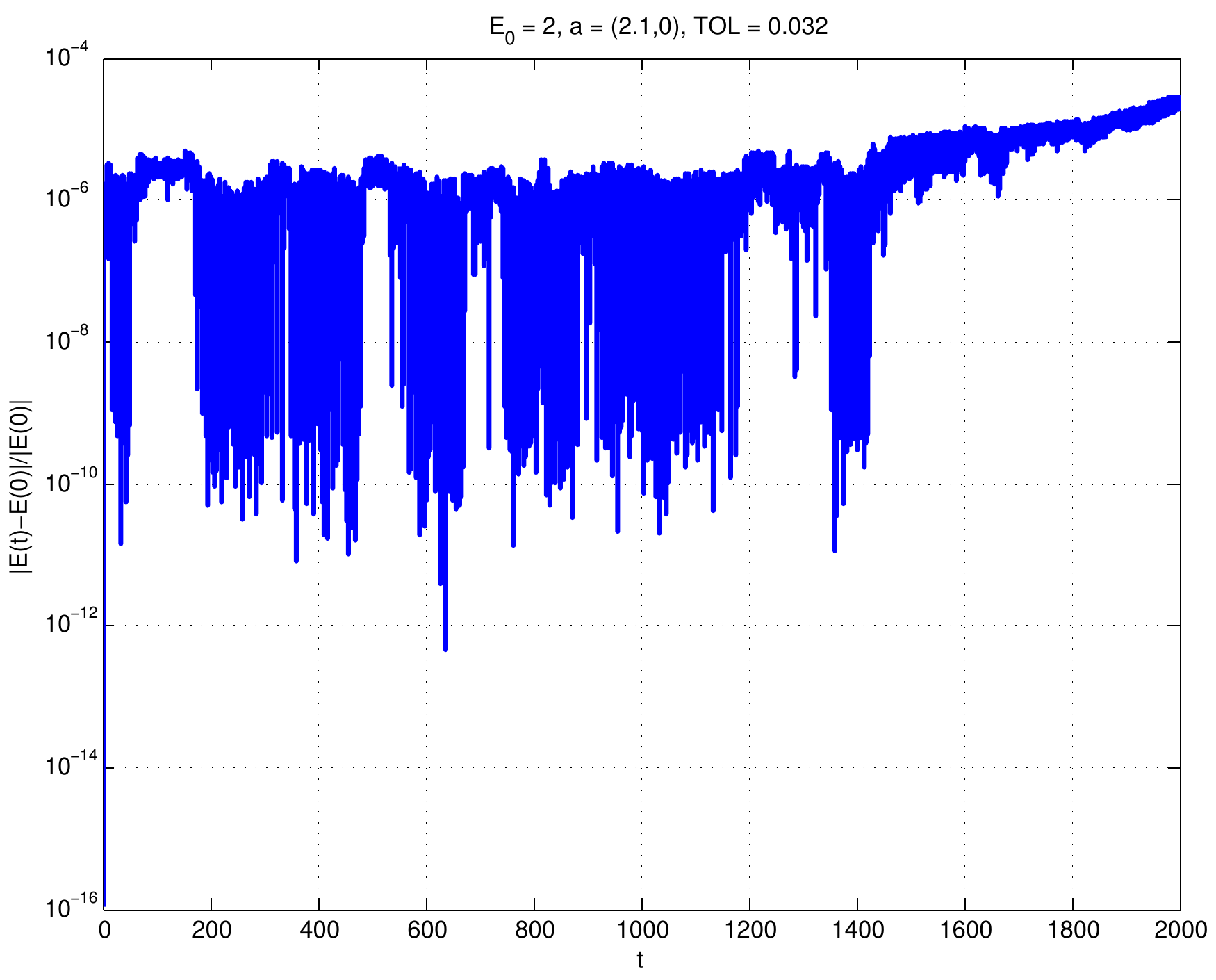} \\
\includegraphics[width=0.45\textwidth]{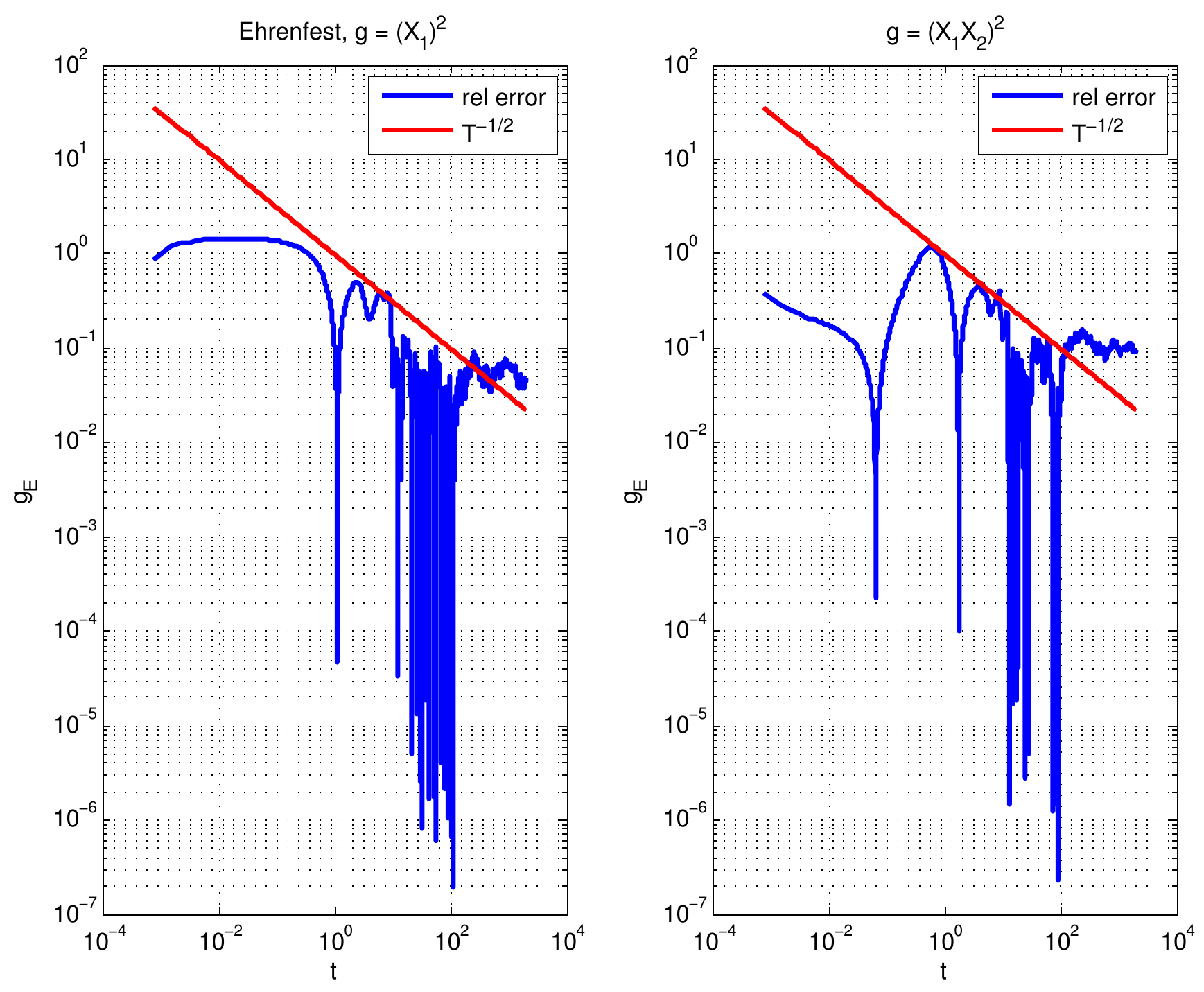}
\includegraphics[width=0.45\textwidth]{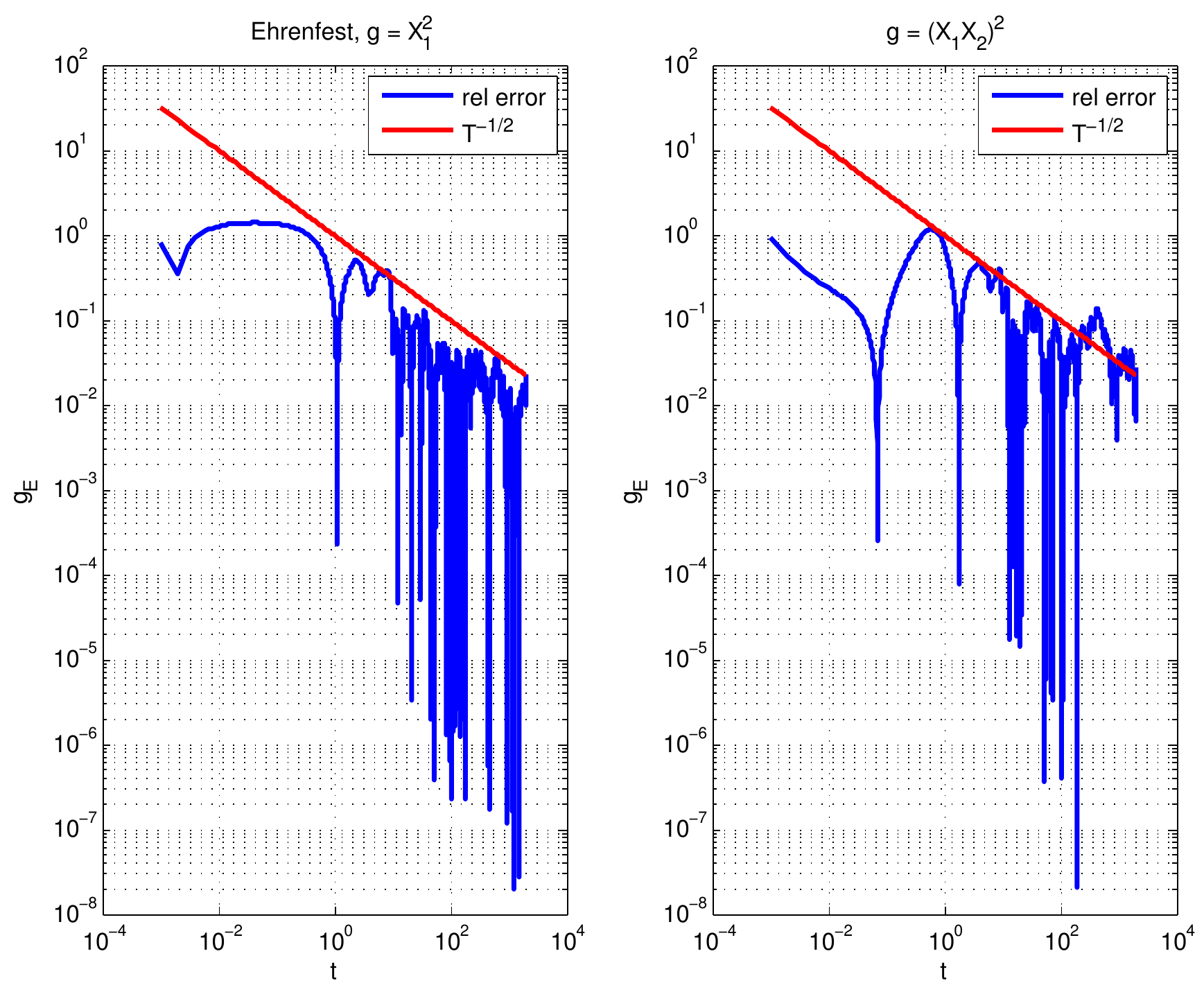}
\caption{Solution plots for uniform (left column) and
adaptive mass (right column) cases for smaller mass parameters. The conical intersection is at $a = (2.1,0)$,
the mass parameter is $M=1600$ for the uniform mass case, and the parameter $\epsilon = 0.032$ for the adaptive
mass case. The required number of time-steps for the time interval $[0,2000]$ is $2.5\times 10^6$ for the uniform mass case,
and $2.2\times 10^6$ 
for the adaptive mass case in the 2D-problem \eqref{eq:2d-potential}.
The first row shows that the excitation probabilities, $p_E$, are much smaller when adaptive mass is used,
the second row shows that the energy is better preserved in the adaptive mass case, and the third row
shows that Ehrenfest observable approximation can not converge to the ergodic limit as
$T\rightarrow\infty$, because of the drift away from the ground state. Since the adaptive mass 
algorithm approximately preserves the ground state better than uniform mass, for finite time
intervals, it can also approximate observables better.
}\label{fig:small-mass-sol}
\end{figure}

This numerical example illustrates  the efficiency obtained by determining 
the artificial mass parameter adaptively. The systematic
choice of the mass parameter also reduces the experimental complexity.
\subsection{An example with smaller spectral gap}
We consider the case with $a = (2.1,0)$ closer to the classically allowed region,  in which the excitation probabilities, $p_E$,
are larger than the case with $a = (2.25,0)$ considered in Section \ref{sec:2d-prob}. 
Figure~\ref{fig:small-mass-sol} compares uniform mass ratio to adaptive mass ratio regarding excitation probability, energy conservation and accuracy of observables.
In this example, we choose 
the mass $M=1600$ for the uniform mass case, and the parameter $\epsilon = 0.032$ for the adaptive
mass case. The number of time-steps required for the uniform mass case is $2.5\times 10^6$ and for the
adaptive mass case it is $2.2\times 10^6$; 
consequently the work in the
two settings are approximately the same 
while the adaptive mass yields a more accurate solution with smaller excitation probability,
better energy conservation, and higher accuracy of observables.
\subsection{A one dimensional problem}
We consider the one dimensional, time-independent Schr\"odinger equation~\eqref{eq:schrod}
with the heavy-particle coordinate~$X \in \mathbb{R}$, the two-state light-particle 
coordinate~$x \in \{x_-,x_+\}$, and the potential matrix
\begin{equation}
\VOPER(X) := \left[\begin{array}{cc}
                               X\cos 2X+\delta\sin 2X-1+r(X) & -X\sin 2X+\delta\cos 2X \\
                               -X\sin 2X+\delta\cos 2X & -(X\cos 2X+\delta\sin 2X)-1+r(X)
                        \end{array}\right], \label{eq:potential-1d}
\end{equation}
where the parameter $\delta$ is a non-negative constant, and the function~$r:\mathbb R \rightarrow \mathbb R$ is given by
$$r(X) := \left\{\begin{array}{ll}
                                (X+2)^2 & \mbox{if $X<-2$}\COMMA \\
                               (X-2)^2 & \mbox{if $X>2$}\COMMA \\
			0 & \mbox{otherwise}.
                        \end{array}\right.$$
The eigenvalues of $V$ are given by $\lambda_\pm(X) = r(X)-1 \pm \sqrt{X^2+\delta^2}$, which gives a minimum 
distance between $\lambda_-$ and $\lambda_+$ of size $2\delta$ at $X=0$.
We choose the energy $E=1$. We note that a larger value of $\delta$ will give a smaller excitation probability,
whereas a smaller value of $\delta$ will give a larger excitation probability.

We use the RATTLE alorithm~\cite{shake_rattle} for approximation of the Car-Parrinello molecular dynamics~\eqref{eq:cp} simulations, and the
St\"ormer-Verlet method~\cite{verlet-method} for the Ehrenfest molecular dynamics~\eqref{eq:mod-ehrenfest} simulations.
We choose the initial data $X_0=-4$,
$P_0 = \sqrt{2(E-\langle\psi_0,V(X_0)\psi_0\rangle)}$, and $\psi_0$ to be the ground state eigenvector of $V(X_0)$.

\begin{figure}[h!]
\centering
\includegraphics[width=0.45\textwidth]{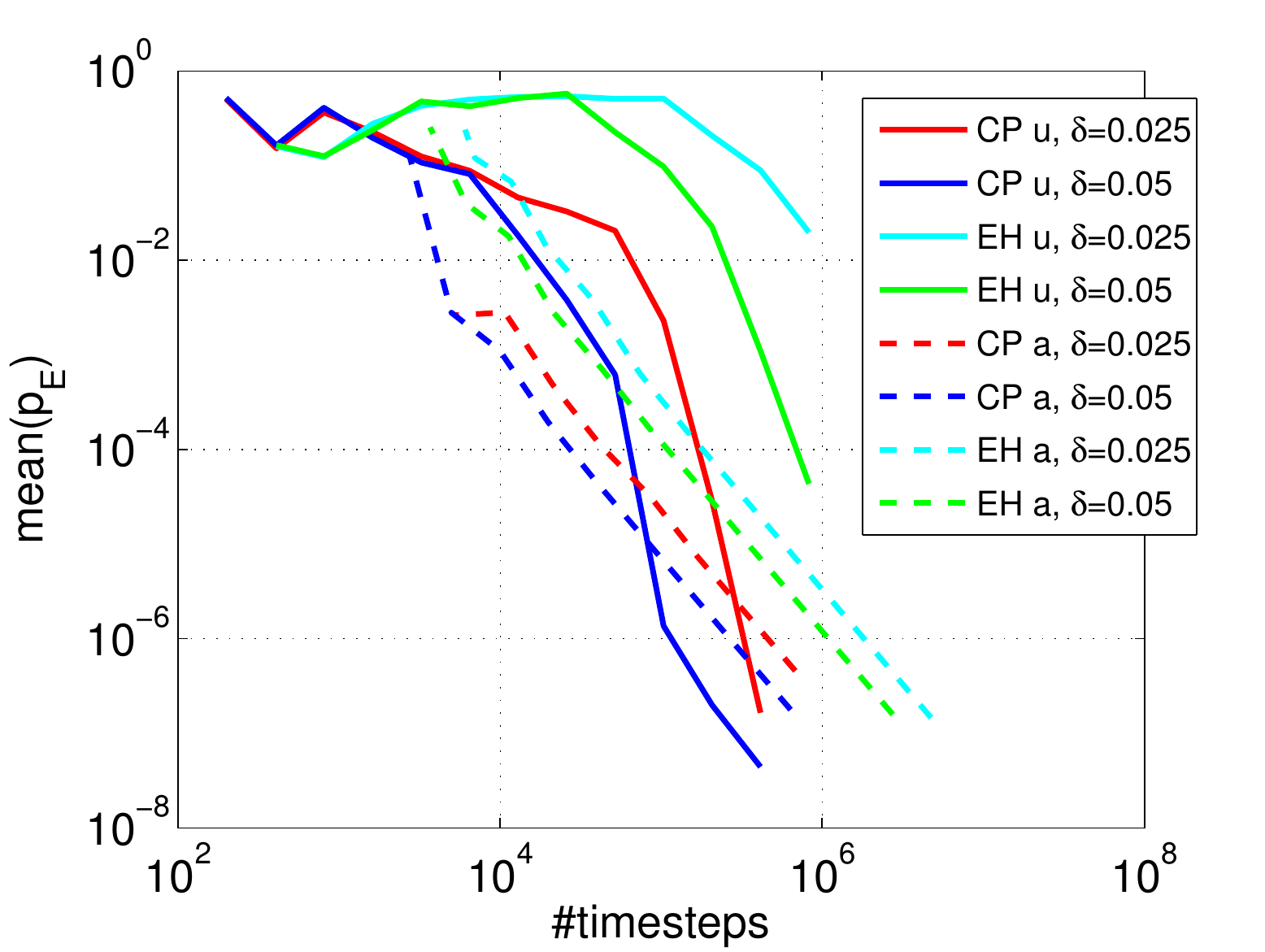}
\caption{
Plot showing the arithmetic mean of the excitation probabilities, $p_E$, as a function of computational
work (number of time-steps) for different spectral gaps $\delta$. The Car-Parrinello (CP) and
Ehrenfest (EH) have parameter $\gamma_{CP}=2$ and $\gamma_E=4$ in equation \eqref{eq:potential-1d}. 
Solid curves correspond to uniform mass (u), and dotted lines to adaptive mass (a). We see that
the Car-Parrinello molecular dynamics performs better than the Ehrenfest molecular
dynamics, in the sense of computational work needed for a given accuracy.
}\label{fig:1D-delta-pE}
\end{figure}

\begin{figure}[h!]
\centering
\includegraphics[width=0.45\textwidth]{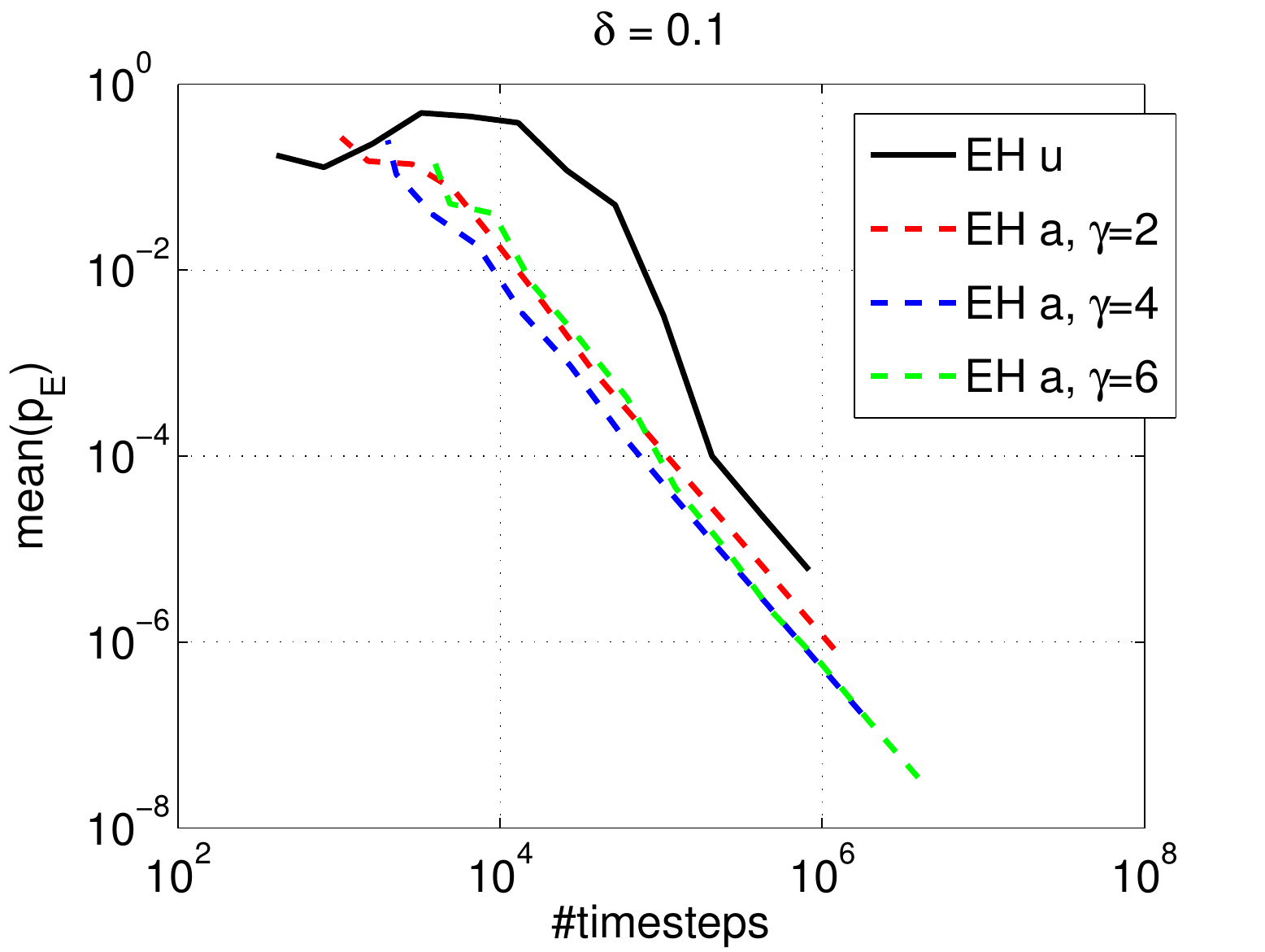}
\caption{Plot showing the arithmetic mean of the excitation probability, $p_E$, as a function of the number of
time-steps with different $\gamma_E$ for Ehrenfest molecular
dynamics in the 1D-case \eqref{eq:potential-1d}.}\label{fig:EH-variableExpn}
\end{figure}

\begin{figure}[h!]
\centering
\includegraphics[width=0.45\textwidth]{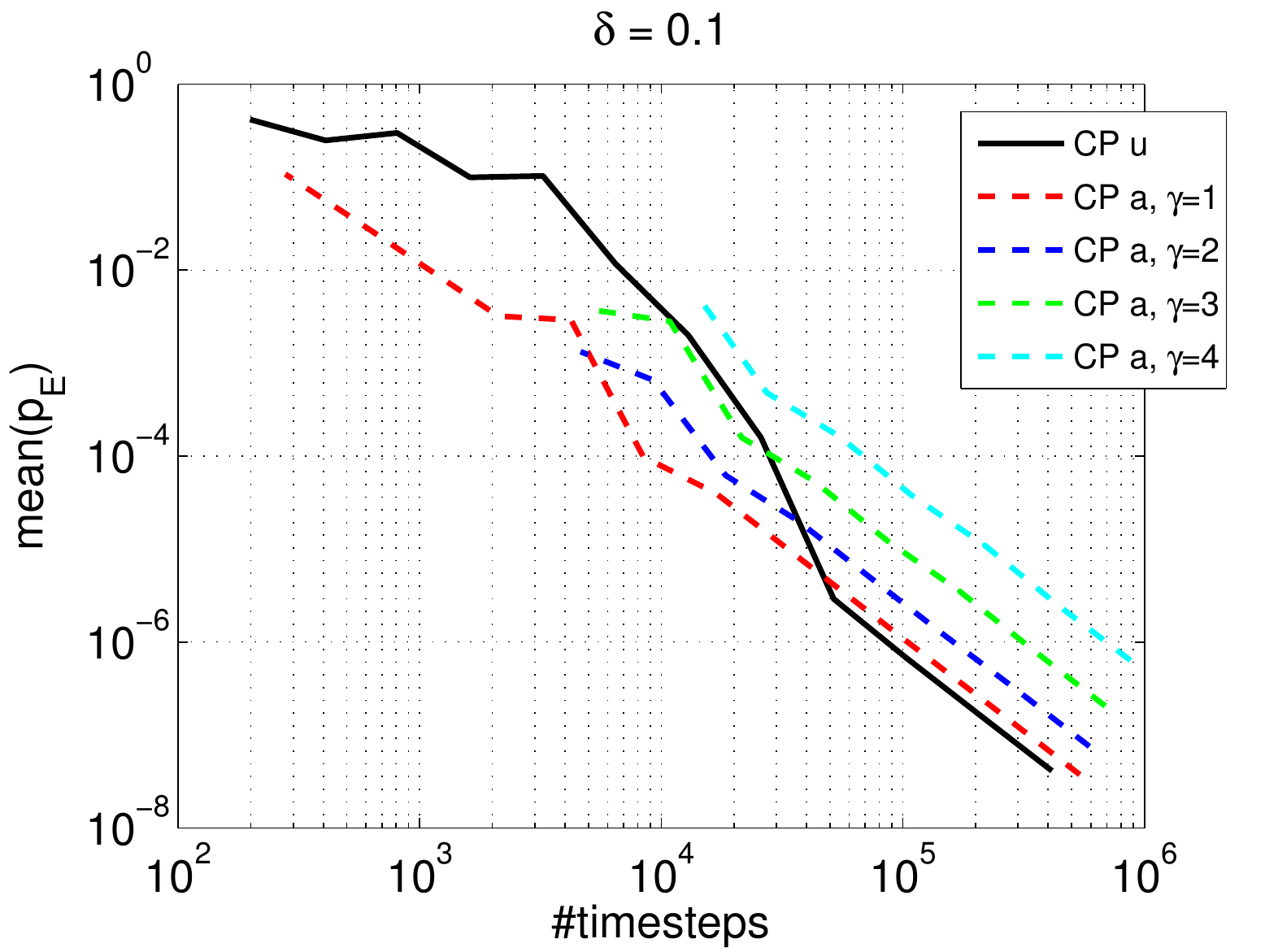}
\caption{Plot showing the arithmetic mean of the excitation probability, $p_E$, as a function of the number of 
time-steps with different $\gamma_{CP}$ for Car-Parrinello molecular
dynamics in the 1D-case \eqref{eq:potential-1d}.}\label{fig:CP-variableExpn}
\end{figure}

Figure~\ref{fig:1D-delta-pE} shows that Car-Parrinello molecular dynamics
performs better than Ehrenfest molecular dynamics, in the sense of computational
work needed to achieve a desired accuracy.
Observe that for very small excitation probability, using large mass and very 
many time-steps, the uniform case is somewhat more efficient, while the adaptive
method is more efficient for less computational work. 
Figure~\ref{fig:mean-pE}, for the two dimensional case, shows that adaptive Ehrenfest
method  is more efficient than uniform Ehrenfest method.
Figure~\ref{fig:EH-variableExpn} shows that for Ehrenfest dynamics the choice 
$\gamma_E=4$ yields better results in comparison to $\gamma_E=2$ or $\gamma_E=6$, and 
Figure~\ref{fig:CP-variableExpn} shows that for  Car-Parrinello dynamics
$\gamma_{CP}=1$ yields the best result.
\begin{figure}[h!]
\centering
\includegraphics[width=0.45\textwidth]{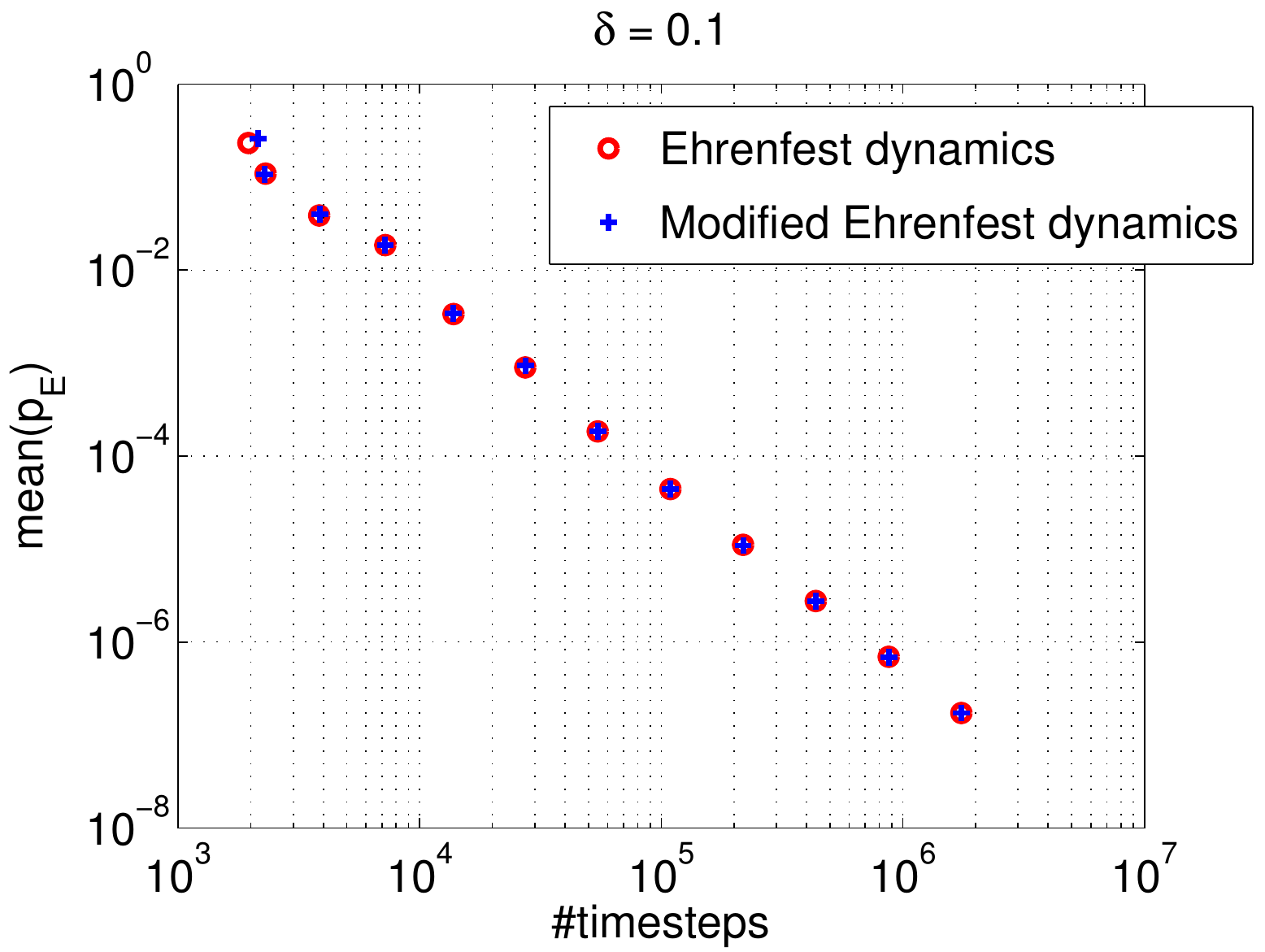}
\caption{Plot comparing the arithmetic  mean of the excitation probability, $p_E$, for
Ehrenfest dynamics~\eqref{eq:ehrenfest} and for the modified version~\eqref{eq:mod-ehrenfest},
as a function of the number of 
time-steps using the algorithm based on random perturbation \eqref{eq:adaptiveM-EH-random} in $\dot\psi$
with $\delta = 0.1,$ $t\in[0,8]$ and $\gamma_E = 4$ in the 1D-case \eqref{eq:potential-1d}.
}\label{fig:EH-randomperturb-meanpE}
\end{figure}
\begin{figure}[h!]
\centering
\includegraphics[width=0.45\textwidth]{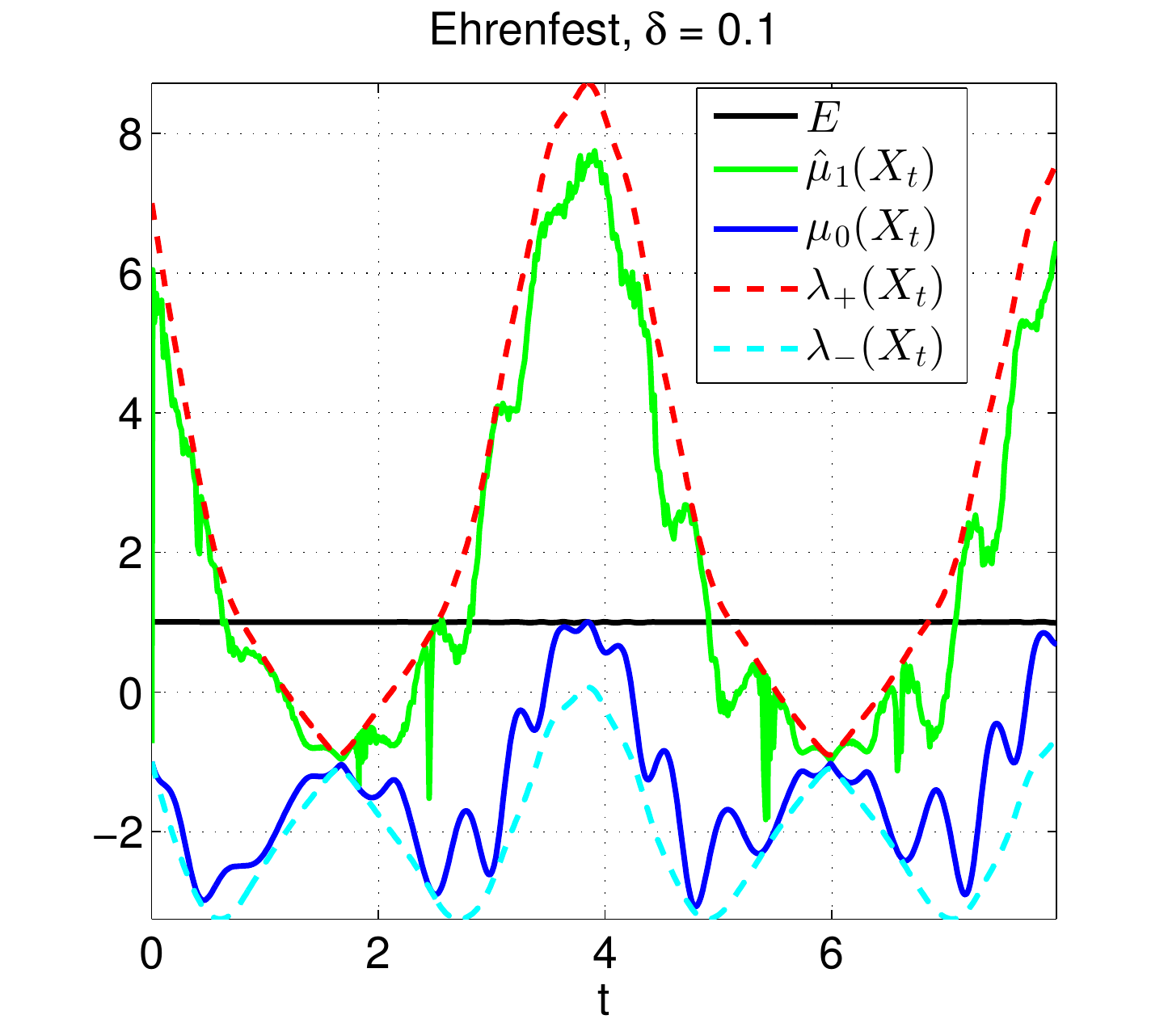}
\includegraphics[width=0.45\textwidth]{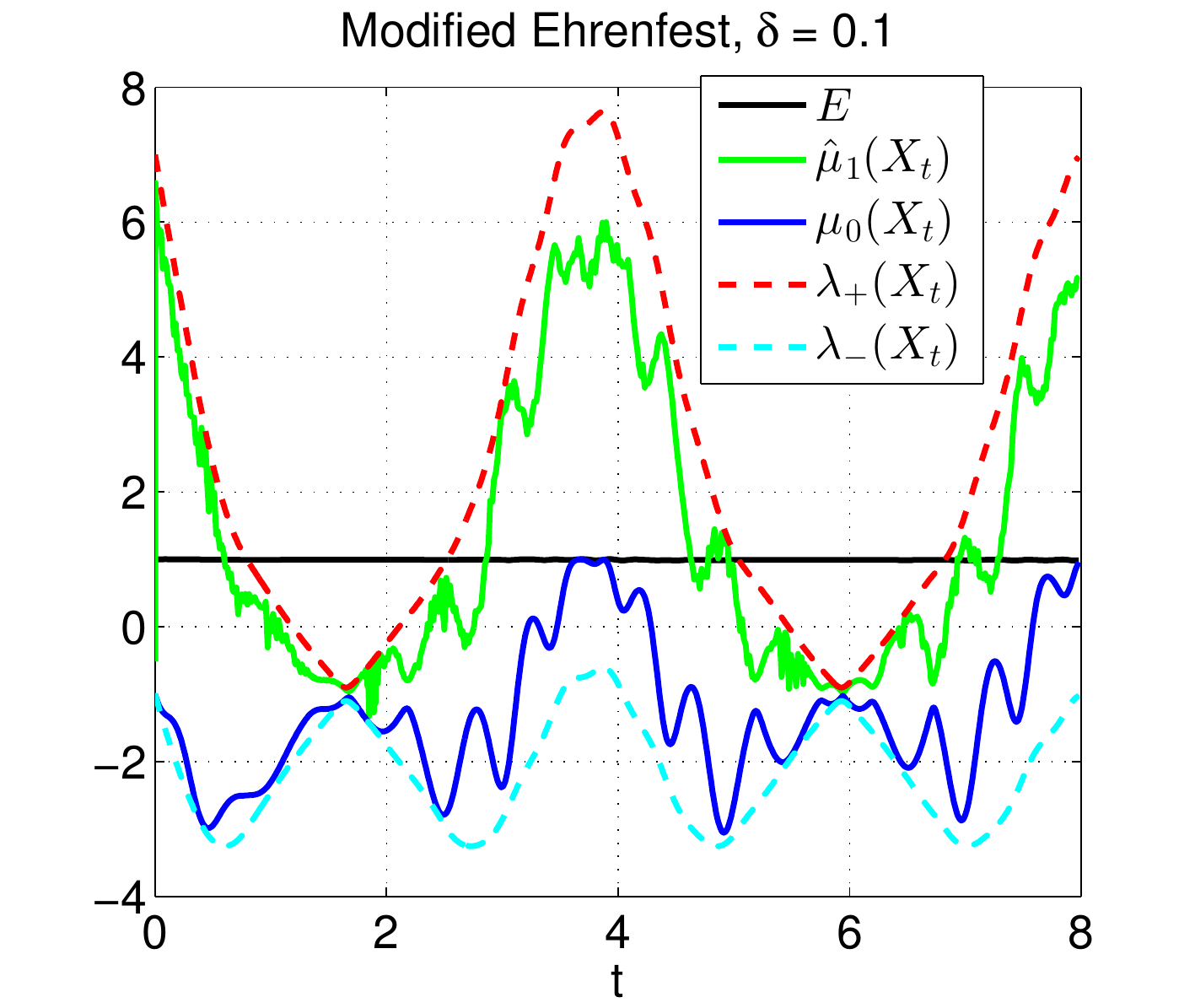} \\
\includegraphics[width=0.45\textwidth]{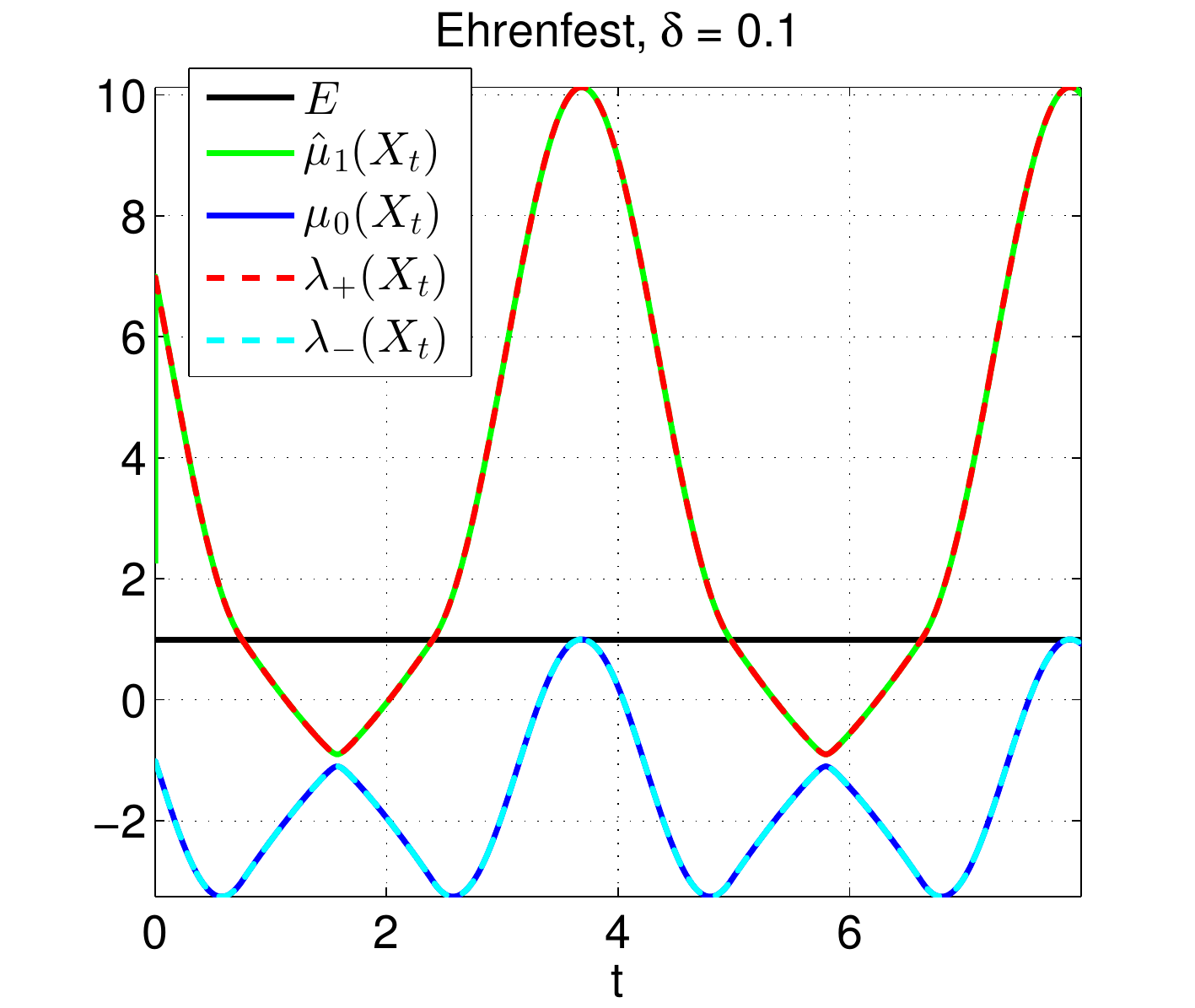}
\includegraphics[width=0.45\textwidth]{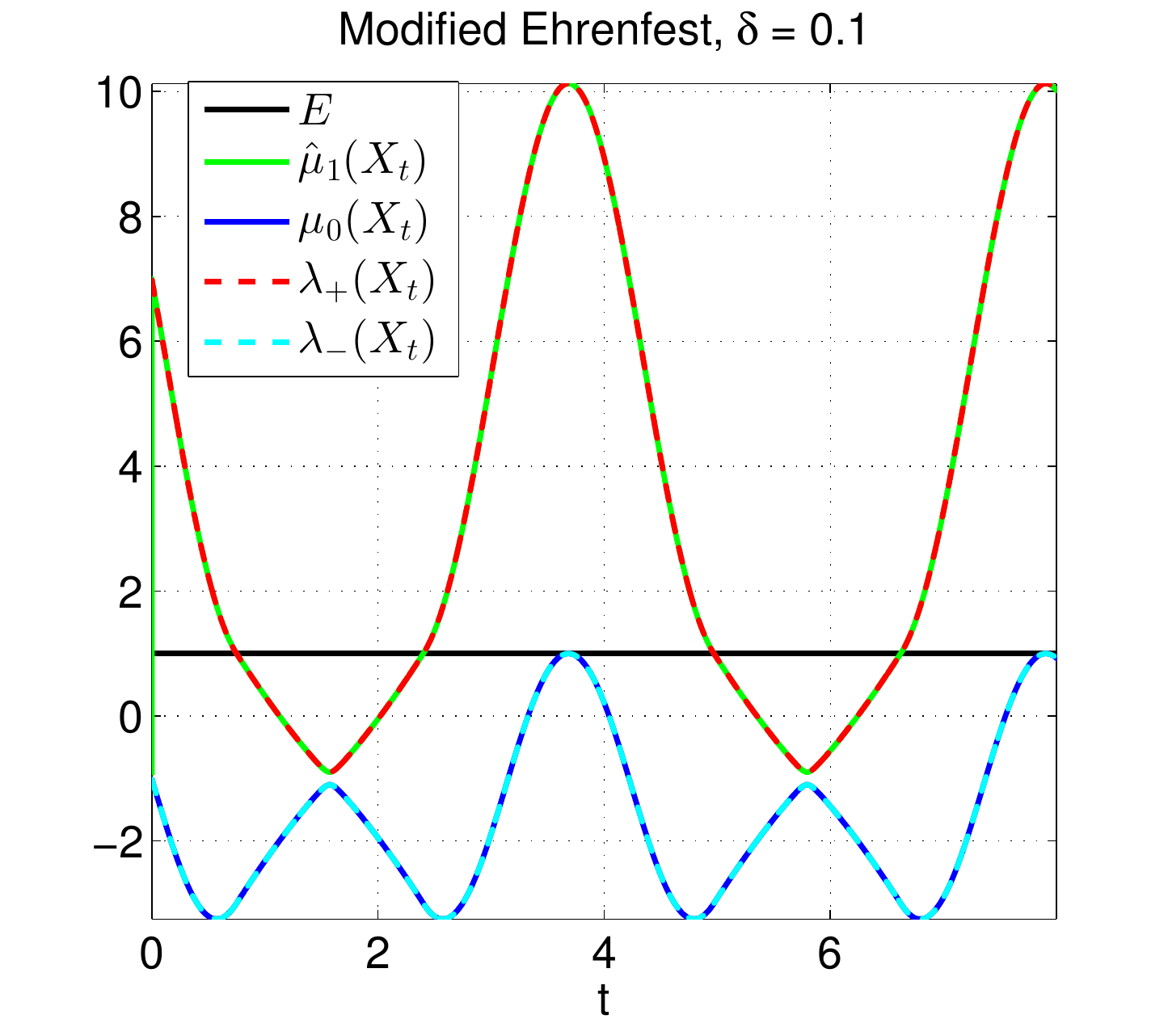}
\caption{Plots showing approximation of the eigenvalues for the algorithm based on random perturbation in $\dot\psi$ for 
(left column) the Ehrenfest   dynamics~\eqref{eq:ehrenfest},
and (right column) the modified version~\eqref{eq:mod-ehrenfest} with $\delta = 0.1,$ $t\in[0,8],$ and $\gamma_E = 4$
in the 1D-case \eqref{eq:potential-1d}.
We choose the constant $\epsilon$ as $0.512$ (first row), $0.001$ (second row).
}\label{fig:EH-randomperturb-energy}
\end{figure}
\begin{figure}[h!]
\centering
\includegraphics[width=0.45\textwidth]{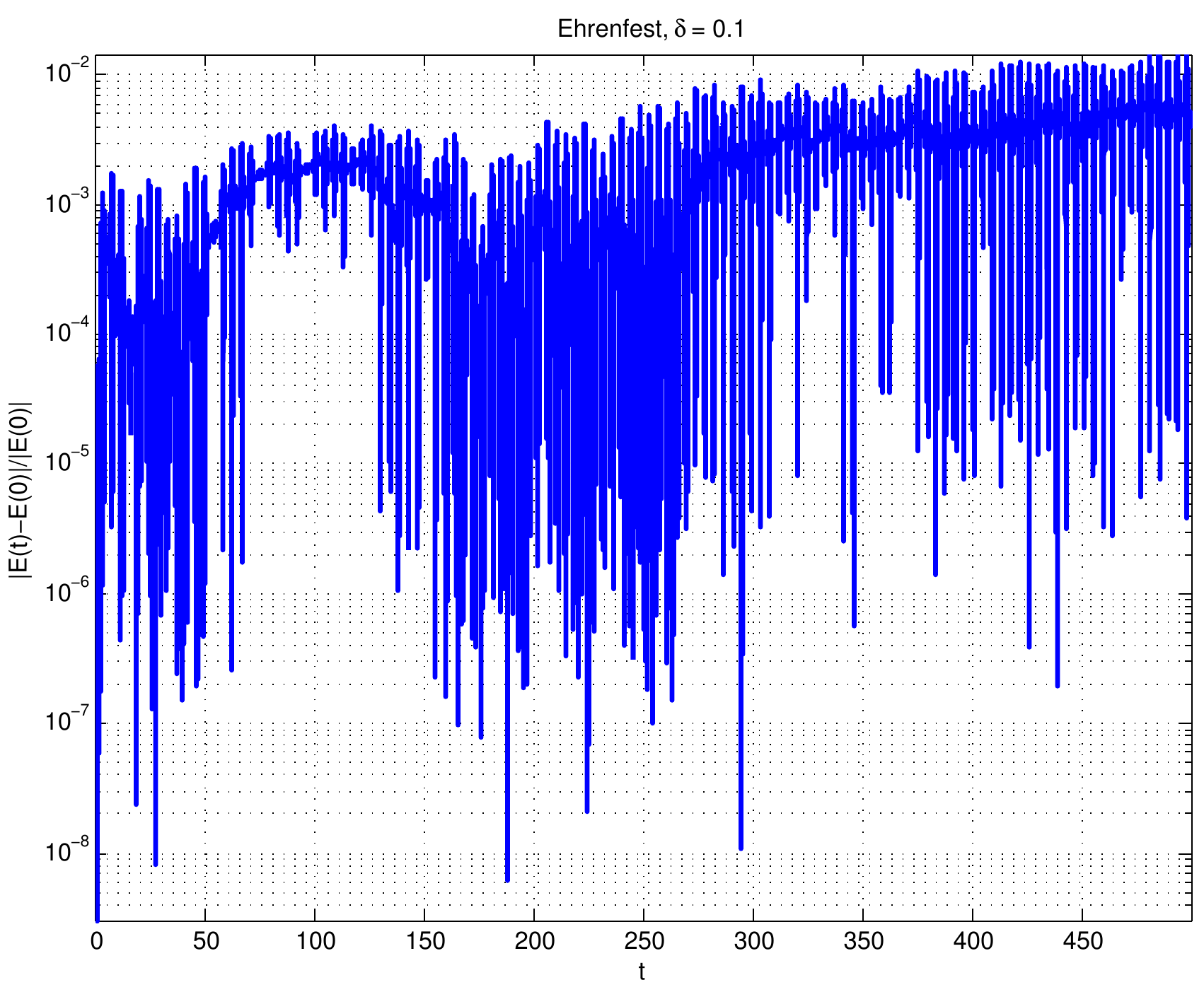}
\includegraphics[width=0.45\textwidth]{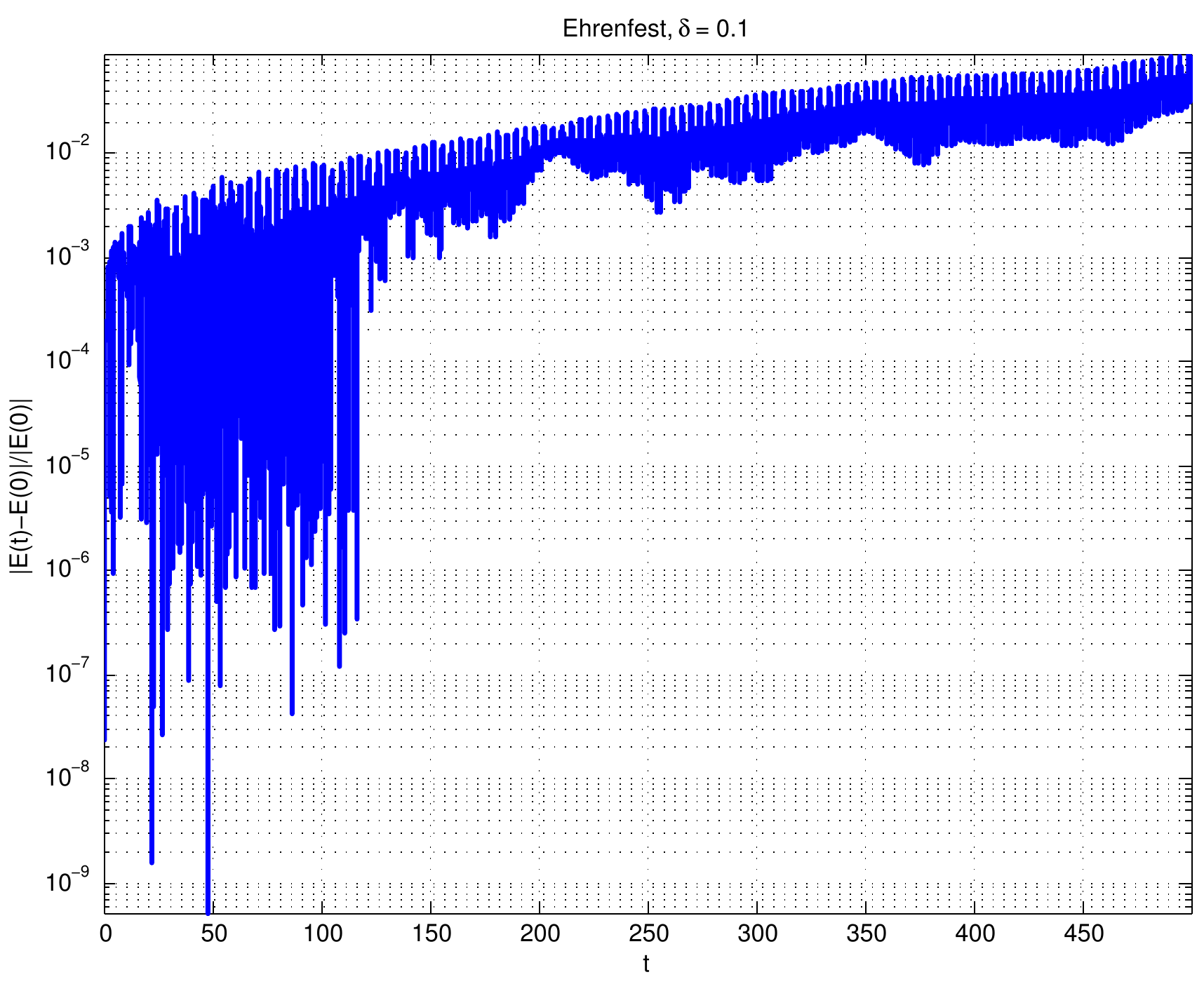} \\
\includegraphics[width=0.45\textwidth]{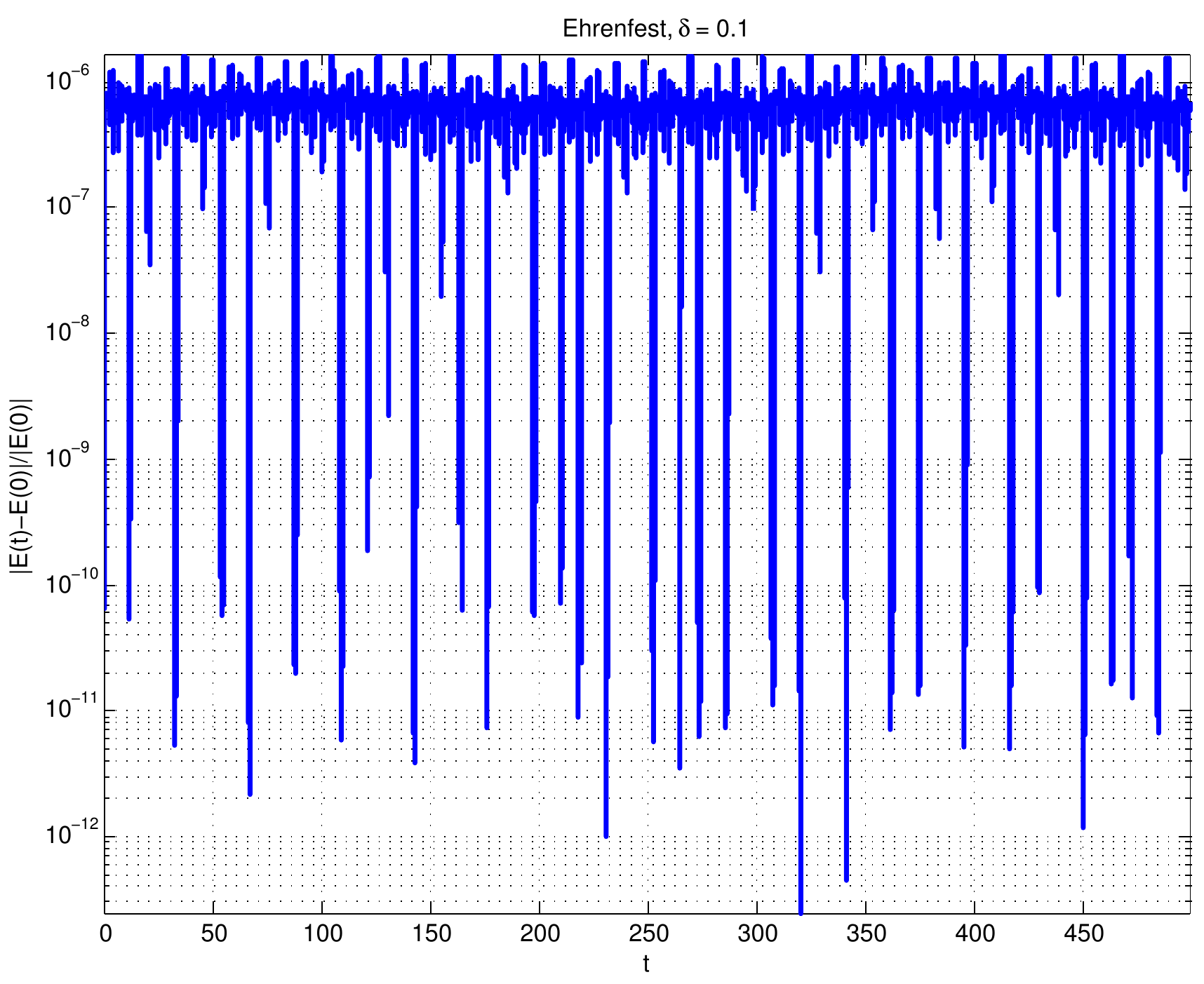}
\includegraphics[width=0.45\textwidth]{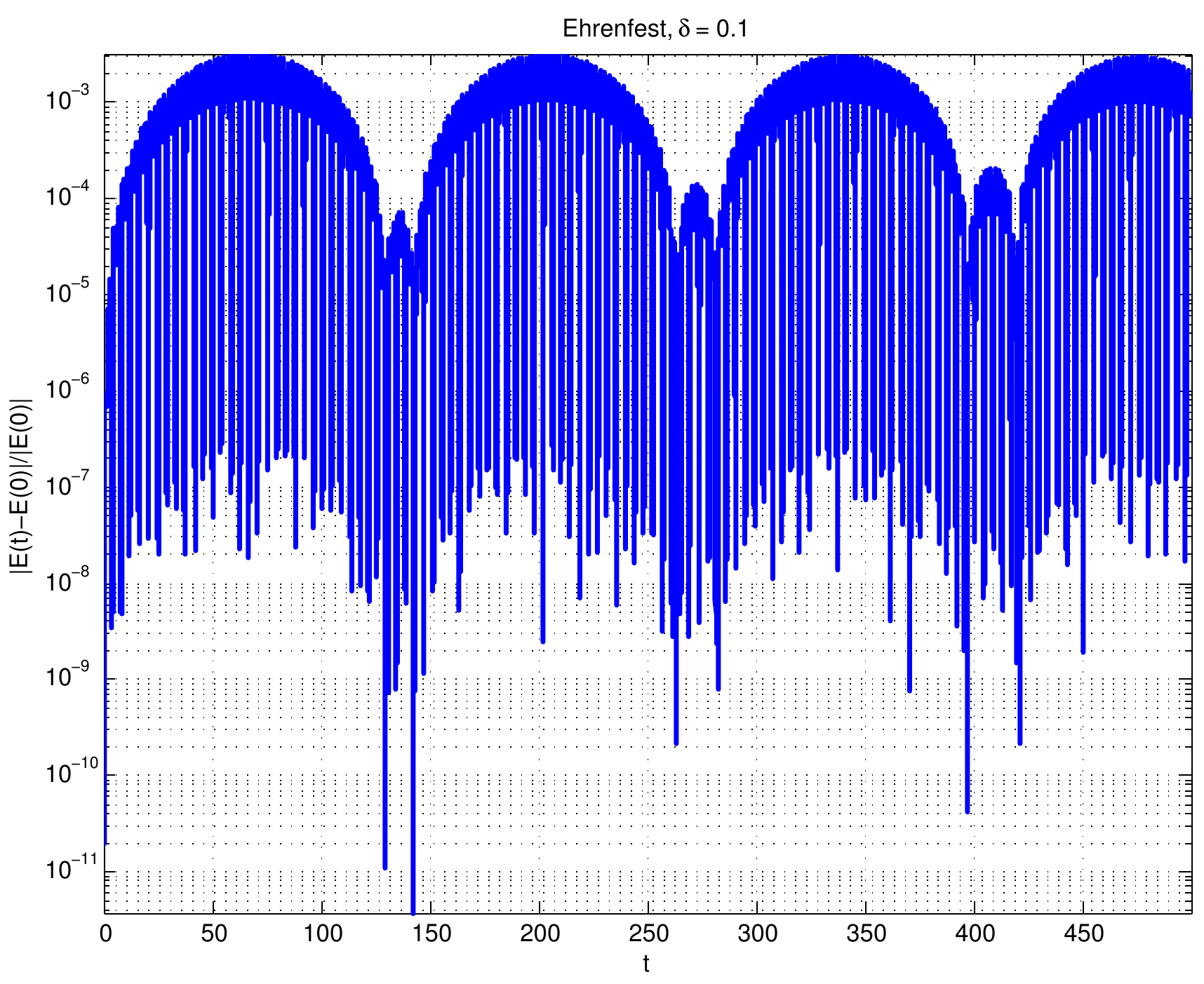}
\caption{Plots showing the energy deviation for the algorithm based on random perturbation \eqref{eq:adaptiveM-EH-random} in $\dot\psi$ for 
(left column) the Ehrenfest  dynamics~\eqref{eq:ehrenfest},
and (right column) the modified version~\eqref{eq:mod-ehrenfest} with $\delta = 0.1,$ $t\in[0,500],$ and $\gamma_E = 4$ in the 1D-case \eqref{eq:potential-1d}.
We choose the constant $\epsilon$ as $0.064$ (first row), $0.004$ (second row).
}\label{fig:EH-randomperturb-energy-deviation}
\end{figure}
\subsection{An alternative adaptive mass algorithm}\label{sec:alt}
We test an alternative criterion for the time-dependent mass parameter $M_E(t)$
for both the standard Ehrenfest molecular dynamics~\eqref{eq:ehrenfest} and the modified translated Ehrenfest
dynamics model~\eqref{eq:mod-ehrenfest} by choosing
\begin{equation}\label{eq:adaptiveM-EH-random}
\hat M_E(t) := \epsilon^{-2}\max\left(1, (\frac{|P_t|^{1/2}}{|\hat\mu_1(X_t)-\mu_0(X_t)|})^{\gamma_E}\right)\COMMA
\end{equation}
where
\[
\hat\mu_1(X_t) := \frac{\langle\hat\psi_t, V(X_t)\hat\psi_t\rangle}{\langle \hat\psi_t, \hat\psi_t\rangle}\COMMA
\]
with
\begin{equation*}
\hat\psi_t := \dot\psi_t + \frac{\mathbb{\xi}}{\max(\hat M_E(t-\Delta t),1/\epsilon^2)}\COMMA
\end{equation*}
and
\begin{equation*}
\hat\psi_0 := \dot\psi_0 + \mathbb{\xi}\epsilon^2\COMMA
\end{equation*}
where $\mathbb{\xi}$ is a random vector of the same dimension as $\dot\psi$
with independent components that are uniformly distributed random numbers on the interval $(0,1)$ and
$\Delta t$ is the size of the time-steps. The motivation for this test is that without the random perturbations the standard method \eqref{eq:ehrenfest} exhibits numerical roundoff problems and  
the random perturbations seems to cure this problem, at the expense of complicating the algorithm.

Figure~\ref{fig:EH-randomperturb-meanpE} shows that the computed mean value of the excitation probabilities, $p_E$, coincide 
for the Ehrenfest  molecular dynamics~\eqref{eq:ehrenfest} and the modified verstion~\eqref{eq:mod-ehrenfest}.
Figure~\ref{fig:EH-randomperturb-energy} illustrates the approximations of the eigenvalues and 
Figure~\ref{fig:EH-randomperturb-energy-deviation} determines the energy deviation over time. The result is that 
the Ehrenfest dynamics model~\eqref{eq:ehrenfest} preserves energy better than the modified version \eqref{eq:mod-ehrenfest}.
\section{Conclusions}

In this work, we show by numerical examples in two simple model problems that the
computational efficiency of Ehrenfest molecular dynamics and Car-Parrinello molecular
dynamics can be simplified and improved using adaptive mass ratios, in particular, for
the cases of small spectral gaps in avoided crossings of electron potential surfaces.
We use the Landau-Zener probability and its generalization to Ehrenfest dynamics
to motivate the choice of the adaptive mass.
A two-dimensional numerical example shows that the use of time-dependent adaptive mass in Ehrenfest 
molecular dynamics gives improved computational efficiency over the use of uniform mass. 
We also observable that Car-Parrinello dynamics is more efficient than Ehrenfest dynamics for a one dimensional test case.


\section{Appendix: The Ehrenfest approximation}\label{sec:bo}
The purpose of this section is to analyze how well Ehrenfest solutions approximate
Born-Oppenheimer solutions,  in
the sense that the amplitude function, $\psi$, in the Ehrenfest dynamics is close to the electron ground state, $\Psi_0$,
when a spectral gap condition holds.
%
The electron eigenvalues $\{\lambda_k\}$
satisfy the spectral gap condition provided there is  a positive $\delta$ such that
\begin{equation}\label{gap_c1}
\begin{split}
   &\inf_{k\ne 0,\ Y\in D} |\lambda_k(Y)-\lambda_0(Y)|>\delta\COMMA\\
&\sup_{ Y\in D} |\nabla\Psi_0(Y)|<\delta^{-1}\COMMA   
\end{split}
\end{equation}
where $D:=\{X_t \SEP t\ge 0\}$
is the set of all nuclei positions
obtained for the Ehrenfest dynamics  \eqref{eq:ehrenfest} with initial energy $|P_0|^2/2+\langle \psi_0,V(X_0)\psi_0\rangle/\langle \psi_0,\psi_0\rangle=E$,
for all considered initial data. 
If $(X_t,P_t,\psi_t)$ solves the standard Ehrenfest dynamics \eqref{eq:ehrenfest}, then
$t\mapsto(X_t,P_t,\psi_te^{iM^{1/2}\int_0^t\lambda_0(X_s)ds})$ solves the translated Ehrenfest
dynamics 
\begin{equation}\label{ehren_trans}
\begin{split}
\ddot X_t&=-\nabla\lambda_0(X_t) -\frac{\langle\psi_t,\nabla\tilde V(X_t)\psi_t\rangle}{\langle\psi_t,\psi_t\rangle}\, ,\\
i\dot\psi_t&=M^{1/2}\tilde V(X_t)\psi_t\, .
\end{split}
\end{equation}
The following result shows an estimate of  transition probability $p_E=|\psi^\perp|^2=\mathcal O(M^{-1})$ 
for this case with a spectral gap.


\begin{lemma}\label{born_oppen_lemma}
  Assume that $(X,\psi)$  satisfies \eqref{ehren_trans} with initial condition $\psi^\perp(0)=0$,
  the eigenvalue pairs $\{\lambda_k(X),\Psi_k(X)\}_{k=0}^\infty$ of $V(X)$ are smooth functions
  of $X$, and
  the spectral gap condition \eqref{gap_c1} holds.
  Then for each  $t\in\rset_+$  the orthogonal decomposition 
  $\tpsi_t=\psi_0(t) \oplus\psi^\perp(t)$ in $\mathbb C^n$,
  where $\psi_0(t)=\alpha_t\Psi_{\BO}(X_t)$
  for some $\alpha_t\in\mathbb C$, satisfies
  \begin{equation}\label{omega_estimate}
    \begin{split}
    |\psi^\perp(t)| &\le CM^{-1/2}\\
    \end{split}
  \end{equation}
for a constant $C=\mathcal O(\delta_*^{-2})$ as 
$\delta_*:=\min_{t\in\rset_+}|\lambda_1(X_t)-\lambda_0(X_t)||P_t|^{-1/2}\rightarrow 0+$. 
\end{lemma}

\begin{proof}
  We consider the decomposition
  $\tpsi=\psi_0\oplus\psi^\perp$,
  where $\psi_0(\tau)$ is an eigenvector of $\VOPER(X_\tau)$, 
  satisfying $\VOPER(X_\tau)\psi_0(\tau)=\lambda_0(\tau)\psi_0(\tau)$
  for the eigenvalue $\lambda_0(\tau)\in\rset$. 
  This {\it ansatz} is motivated by the zero residual
  \begin{equation}\label{R_residual}
     \ROPER\tpsi:= \dot\tpsi +\Iunit M^{1/2}\tilde \VOPER\tpsi= 0 
  \end{equation}
  and projection of the eigenvector
 \[\begin{split} %
     \LPROD{\PSHARP{(\dot{\psi}_0)}}{\psi_0} &=0 \\ 
     M^{1/2}\tilde \VOPER\psi_0                  &= 0\COMMA 
  \end{split}\] %
  where 
  \begin{equation}\label{natural}
      w(X)=\LPROD{\Psi_{\BO}(X)}{w(X)} \Psi_{\BO}(X)\oplus \PSHARP{w(X)}
  \end{equation}
  denotes the orthogonal decomposition in the eigenvector direction $\Psi_{\BO}(X)$ and
  its orthogonal complement in $\mathbb C^n$, i.e. $\langle\Psi_{\BO}(X),\PSHARP w(X)\rangle =0$. We consider  the linear operator $\ROPER$ 
  in \eqref{R_residual}. 
  The orthogonal splitting  $\tpsi=\psi_0\oplus\psi^\perp$ and the projection 
  $\PSHARP{(\cdot)}$ in \eqref{natural} imply
  \[
    \begin{split}
        0 &= \PSHARP{\left(\ROPER(\psi_0+\psi^\bot)\right)}\\
          &= \PSHARP{\left(\ROPER(\psi_0)\right)} + \PSHARP{\left(\ROPER(\psi^\bot)\right)}\\
          &= \PSHARP{(\ROPER\psi_0)} + \PSHARP(\frac{d}{dt}{\psi}^\bot) +\Iunit M^{1/2}\tilde\VOPER\psi^\bot 
          \COMMA
    \end{split}
  \] 
  where the last step  follows from the orthogonal splitting 
  \[
      \PSHARP{(\tilde\VOPER\psi^\bot)}=\tilde\VOPER\psi^\bot\, .
  \]
    Let $\widetilde \SOPER_{\tau,\sigma}$ be the solution operator from time $\sigma$ to $\tau$, 
    for the generator 
  \[
     v\mapsto \Iunit M^{1/2}\tilde Vv
  \]
   on the manifold generated by $\Pi$.
   Then  the perturbation $\psi^\bot$ can be determined from the projected residual
  \[
    \Pi(t)\frac{d}{dt}\big(\psi^\bot(t)\big) 
    =-\Iunit M^{1/2} \tilde \VOPER \psi^\bot(t) - \PSHARP{(\ROPER\psi_0(t))}\, .
  \]
  and we have the solution representation 
  \begin{equation}\label{r_int}
    \psi^\bot(\tau)= \widetilde \SOPER_{\tau,0} \psi^\bot(0)
                       -\int_0^\tau  \widetilde \SOPER_{\tau,\sigma} 
                         \PSHARP{\left(\ROPER\psi_0(\sigma)\right)}\, d\sigma\, ,
  \end{equation}
  as motivated in Remark \ref{second_order_change}.
    
  Integration by parts in \eqref{r_int} introduces the factor $M^{-1/2}$ we seek
  \begin{equation}\label{s_int_def}
  \begin{split}
          \int_0^\tau \widetilde \SOPER_{\tau,\sigma} \PSHARP\ROPER{\psi_0(\sigma)}\, d\sigma&=
          \int_0^\tau \Iunit M^{-1/2}\frac{d}{d\sigma} 
                      (\widetilde \SOPER_{\tau,\sigma})\tilde \VOPER^{-1}  \PSHARP\ROPER{\psi_0(\sigma)}\, d\sigma\\
             &=\int_0^\tau \Iunit M^{-1/2}\frac{d}{d\sigma} \left(\widetilde \SOPER_{\tau,\sigma}\tilde \VOPER^{-1}  
                    \PSHARP\ROPER{\psi_0(\sigma)}\right)\, d\sigma\\
             &\qquad -\int_0^\tau \Iunit M^{-1/2} \widetilde \SOPER_{\tau,\sigma}\frac{d}{d\sigma}
                 \left(\tilde \VOPER^{-1}(X_\sigma)  \PSHARP\ROPER{\psi_0(\sigma)}\right)\, d\sigma \\
             &= \Iunit M^{-1/2} \tilde \VOPER^{-1}  \PSHARP\ROPER{\psi_0(\tau)}
                - \Iunit M^{-1/2} \widetilde \SOPER_{\tau,0}\tilde \VOPER^{-1}  \PSHARP\ROPER{\psi_0(0)}\\
             &\qquad -\int_0^t \Iunit M^{-1/2} \widetilde \SOPER_{\tau,\sigma}\frac{d}{d\sigma}
                \left(\tilde \VOPER^{-1}(X_\sigma)  \PSHARP\ROPER{\psi_0(\sigma)}\right)\, d\sigma\PERIOD\\
  \end{split}
  \end{equation}
 A spectral decomposition in $\mathbb C^n$, based on the electron 
  eigenpairs $\{\lambda_k,\Psi_k\}_{k=1}^\infty$  satisfying 
  $\VOPER\Psi_k=\lambda_k\Psi_k$ and \eqref{omega_estimate}, 
  then implies
  \begin{equation}\label{vtilde}
  \begin{split}
      \tilde\VOPER^{-1}\PSHARP(\ROPER\psi_{0})
             &=
                      (\VOPER-\lambda_0)^{-1}\PSHARP(\ROPER\psi_{0})\\
             &=\sum_{k\ne 0} 
                    (\lambda_k -\lambda_0)^{-1}\Psi_k \LPROD{\PSHARP(\ROPER\psi_{0})}{ \Psi_k}\\
             &=\sum_{k\ne 0}  (\lambda_k -\lambda_0)^{-1} \Psi_k \LPROD{\PSHARP(\dot\psi_{0})}{\Psi_k}\\
             &=\BIGO(\delta_*^{-2})\, ,
  \end{split}
  \end{equation}
  which applied to the integral in the right hand side of \eqref{s_int_def} together with the boundedness of $\widetilde S$
  shows that $|\psi^\bot|=\BIGO(M^{-1/2})$.

  The estimate of the constant $C$  requires an additional idea:
  one can integrate by parts recursively in \eqref{s_int_def} to obtain
\[
  \begin{split}
       \int_0^\tau \widetilde \SOPER_{\tau,\sigma} \PSHARP\ROPER{\psi_0(\sigma)}\, d\sigma
       &=\left[\widetilde \SOPER_{\tau,\sigma}\Big(\BTOPER \tilde\ROPER -\BTOPER \frac{d}{d\sigma}(\BTOPER\tilde\ROPER) + 
               \BTOPER \frac{d}{d\sigma}
               \big(\BTOPER \frac{d}{d\sigma}(\BTOPER\tilde\ROPER)\big)-\ldots\Big)\right]_{\sigma=0}^{\sigma=\tau}\COMMA \\
                      \BTOPER &:= \Iunit M^{-1/2}\tilde \VOPER^{-1}\COMMA\;\;\;
                 \tilde\ROPER := \PSHARP\ROPER{\psi_0(\sigma)}\COMMA
  \end{split} 
\]
so that by \eqref{r_int} we have
  \[
      \psi^\bot(\tau)= \widetilde \SOPER_{\tau,0} \psi^\bot(0)
     - \left[\widetilde \SOPER_{\tau,\sigma}\Big(\BTOPER \tilde\ROPER -\BTOPER \frac{d}{d\sigma}(\BTOPER\tilde\ROPER) 
     + \BTOPER \frac{d}{d\sigma}
               \big(\BTOPER \frac{d}{d\sigma}(\BTOPER\tilde\ROPER)\big)-\ldots\Big)\right]_{\sigma=0}^{\sigma=\tau}\COMMA
\]
 %
               %
%
which can be written
\begin{equation}\label{beta_exp}
  \psi^\bot(\tau) =  -\sum_{n=0}^N \BTOPER_0^n \ROPER_0(\tau) + \widetilde\SOPER_{\tau,0}\sum_{n=0}^N \BTOPER_0^n \ROPER_0(0)+ \mathcal O(M^{-(N+1)/2})\COMMA
\end{equation}
where $\BTOPER_0:=-\Iunit M^{-1/2}\tilde\VOPER^{-1} \tfrac{d}{d\tau}$ and $\ROPER_0:=\Iunit M^{-1/2}\tilde\VOPER^{-1}\tilde\ROPER$.
The leading order term, as $M\rightarrow\infty$, is  
$|\ROPER_0|=\mathcal O(M^{-1/2}\delta_*^{-2})$ which yields $C=\mathcal O(\delta_*^{-2})$.
 
\end{proof}

\begin{remark}\label{second_order_change}
The solution representation \eqref{r_int} requires that $\langle \frac{d}{dt}\psi^\bot(t),\Psi_0(X_t)\rangle=0$, which
we explain here. 
 Let $\PSHARPT{\tau}{}$ denote the projection on the orthogonal complement to the eigenvector $\Psi_0(\tau)$.
The second order change in the subspace projection
  \[
     \psi^\bot(\tau+\Delta\tau)=\PSHARPT{\tau+\Delta\tau}{\left(\psi^\bot(\tau+\Delta \tau)\right)}= 
     \PSHARPT{\tau}{\left(\psi^\bot(\tau+\Delta \tau)\right)} + \BIGO(\Delta \tau^2)
  \]
  yields $\PSHARP{\big(\frac{d}{dt}({\psi^\bot})\big)}=\frac{d}{dt}\psi^\bot$ and hence 
  $\langle \frac{d}{dt}\psi^\bot(t),\Psi_0(X_t)\rangle=\langle \PSHARP(t)\frac{d}{dt}\big({\psi^\bot}(t)\big),\Psi_0(X_t)\rangle=0$.  
  To explain the second order change start with a function $v$ satisfying $\langle v,\Psi_{\BO}(X_\tau)\rangle = 0$
  and $\Psi_{\BO}(X_\sigma)=\Psi_{\BO}(X_\tau) + \mathcal O(\Delta \tau)$ for $\sigma\in [\tau,\tau+\Delta \tau]$
  to obtain
  \[
  \begin{split}
  \Pi(\sigma)\big(\Pi(\tau+\Delta\tau)v-\Pi(\tau)v\big)&= 
  \Pi(\sigma)\Big(\langle v,\Psi_{\BO}(X_{\tau})\rangle\Psi_{\BO}(X_{\tau})
  -  \langle v,\Psi_{\BO}(X_{\tau+\Delta\tau})\rangle\Psi_{\BO}(X_{\tau+\Delta\tau})\Big)\\
  &=\Pi(\sigma) \mathcal O(\Delta\tau^2) 
  + \Pi(\sigma)\Big(\langle v,\mathcal O(\Delta\tau)\rangle \Psi_{\BO}(X_\tau)\Big)\\
  & =\mathcal O(\Delta\tau^2) 
  + \mathcal O(\Delta\tau)\Big(\Psi_{\BO}(X_\tau)-\langle\Psi_{\BO}(X_\tau),\Psi_{\BO}(X_\sigma)\rangle \Psi_{\BO}(X_\sigma)\Big)\\
  &=\mathcal O(\Delta\tau^2).
  \end{split}
  \]
\end{remark}


%
%
%
%
%
%
%
%
%
%
%
%

%
%
%
%
%
%

%
%
%
%
%

\section*{Acknowledgment}
The work was supported by the Swedish Research Council, grant 621-2010-5647,
and the Swedish e-Science Research Center.


\providecommand{\bysame}{\leavevmode\hbox to3em{\hrulefill}\thinspace}
\providecommand{\MR}{\relax\ifhmode\unskip\space\fi MR }
\providecommand{\MRhref}[2]{%
  \href{http://www.ams.org/mathscinet-getitem?mr=#1}{#2}
}
\providecommand{\href}[2]{#2}

\end{document}